\documentclass{article}

\usepackage{times}
\usepackage{amssymb,amsmath,amsthm}
\usepackage{mathrsfs}
\usepackage{epsfig}
\usepackage{color}

\newcommand\eg{{\em e.g.}~}
\newcommand\cf{{\em cf.}~}

\def\a{\mathsf a}
\def\B{\mathscr B}
\def\C{\mathbb C}

\def\G{\mathcal G}
\def\H{\mathcal H}

\def\O{\mathcal O}
\def\P{\mathsf P}
\def\R{\mathbb R}

\def\12{{\textstyle\frac12}}
\def\<{\left\langle}
\def\>{\right\rangle}
\def\({\left(}
\def\){\right)}
\def\[{\left[}
\def\]{\right]}
\def\dom{\mathcal D}
\def\lone{\mathsf{L}^{\:\!\!1}}

\def\linfloc{\mathsf{L}^{\:\!\!\infty}_{\rm loc}}
\def\ltwo{\mathsf{L}^{\:\!\!2}}
\def\linf{\mathsf{L}^{\:\!\!\infty}}

\def\e{\mathop{\mathrm{e}}\nolimits}
\def\d{\mathrm{d}}
\def\tr{\mathop{\mathsf{Tr}}\nolimits}
\def\diag{\mathop{\mathrm{diag}}\nolimits}
\def\im{\mathop{\mathsf{Im}}\nolimits}
\def\re{\mathop{\mathsf{Re}}\nolimits}
\def\rank{\mathop{\mathrm{rank}}\nolimits}

\def\sgn{\mathop{\mathrm{sgn}}\nolimits}

\newtheorem{Theorem}{Theorem}[section]
\newtheorem{Remark}[Theorem]{Remark}
\newtheorem{Lemma}[Theorem]{Lemma}
\newtheorem{Assumption}[Theorem]{Assumption}
\newtheorem{Corollary}[Theorem]{Corollary}
\newtheorem{Proposition}[Theorem]{Proposition}

\begin{document}

%-----------------------------------------------------------------------------------------------
% Title
%-----------------------------------------------------------------------------------------------

\title{{\Large\textbf{Asymptotics near $\boldsymbol{\pm m}$ of the spectral shift function for
Dirac operators with non-constant magnetic fields}}}
  
\author{Rafael Tiedra de Aldecoa}
\date{\small
\begin{quote}
\emph{
\begin{itemize}
\item[] Facultad de Matem\'aticas, Pontificia Universidad Cat\'olica de Chile,\\
Av. Vicu\~na Mackenna 4860, Santiago, Chile
\item[] \emph{E-mail:} rtiedra@mat.puc.cl
\end{itemize}
}
\end{quote}
}
\maketitle

%-----------------------------------------------------------------------------------------------
% Abstract
%-----------------------------------------------------------------------------------------------

\begin{abstract}
We consider a $3$-dimensional Dirac operator $H_0$ with non-constant magnetic field of constant
direction, perturbed by a sign-definite matrix-valued potential $V$ decaying fast enough at
infinity. Then we determine asymptotics, as the energy goes to $+m$ and $-m$, of the spectral
shift function for the pair $(H_0,H_0+V)$. We obtain, as a by-product, a generalised version
of Levinson's Theorem relating the eigenvalues asymptotics of $H_0+V$ near $+m$ and $-m$ to the
scattering phase shift for the pair $(H_0,H_0+V)$.
\end{abstract}

%-----------------------------------------------------------------------------------------------
\newpage
\tableofcontents
\newpage
%-----------------------------------------------------------------------------------------------

%-----------------------------------------------------------------------------------------------
\section{Introduction}
\setcounter{equation}{0}
%-----------------------------------------------------------------------------------------------

It is known \cite{Tha92} that the free Dirac Hamiltonian $H_m$ acting in the Hilbert space
$\H:=\ltwo(\R^3;\C^4)$ is unitarily equivalent to the operator $h(P)\oplus-h(P)$, where
$P:=-i\nabla$ and $\R^3\ni\xi\mapsto h(\xi):=(\xi^2+m^2)^{1/2}$. For this reason, the set
$\{\pm m\}=h\big[(\nabla h)^{-1}(\{0\})\big]$ of critical values of $h$ plays an important
role in spectral analysis and scattering theory for Dirac operators. For instance, one cannot
prove at $\pm m$ the usual limiting absorption principle for operators $H_m+V$, even with $V$ a
regular perturbation of $H_m$, by using standard commutator methods. Both the statements and
the proofs have to be modified (see \eg \cite{BG09,IM99}).

In this paper, we provide a new account on the spectral analysis of Dirac operators at the
critical values by discussing the behaviour at $\pm m$ of the spectral shift function associated
to sign-definite perturbations of Dirac operators with non-constant magnetic fields. Our work is
closely related to \cite{Rai09} where G. D. Raikov treats a similar issue in the case of magnetic
Pauli operators. It can also be considered as a complement of \cite{RT07}, where general properties
of the spectrum of Dirac operators with variable magnetic fields of constant direction and
matrix perturbations are determined. Other related results on the spectrum of $3$-dimensional
magnetic Dirac operators can be found in
\cite{BG87,BS99,BMR93,BR99,EL99,GM01,GM93,Hac93,HNW89,Ivr98,MR03,Rob99,SU09,Tha91}.

Let us describe the content of this paper. We consider a relativistic spin-$\frac12$ particle
evolving in $\R^3$ in presence of a variable magnetic field of constant direction. By virtue of
the Maxwell equations, we may assume with no loss of generality that the magnetic field has the
form
$$
\vec B(x_1,x_2,x_3)=\big(0,0,b(x_1,x_2)\big).
$$
The system is described in $\H$ by the Dirac operator
$$
H_0:=\alpha_1\Pi_1+\alpha_2\Pi_2+\alpha_3P_3+\beta m,
$$
where $\beta\equiv\alpha_0,\alpha_1,\alpha_2,\alpha_3$ are the usual Dirac-Pauli matrices,
$m>0$ is the mass of the particle and $\Pi_j:=-i\partial_j-a_j$ are the generators of the
magnetic translations with a vector potential
$$
\vec a(x_1,x_2,x_3)=\big(a_1(x_1,x_2),a_2(x_1,x_2),0\big)
$$
that satisfies $B=\partial_1a_2-\partial_2a_1$. Since $a_3 =0$, we write $P_3=-i\partial_3$
instead of $\Pi_3$. We assume that the function $b:\R^2\to\R$ is continuous (see Section
\ref{Unperturbed} for details), so that $H_0$, defined on $C^\infty_0(\R^3;\C^4)$, can be
extended uniquely to a selfadjoint operator in $\H$ with domain $\dom(H_0)$ .

Then we consider a bounded positive multiplication operator
$V\in C\big(\R^3;\B_{\sf h}(\C^4)\big)$, where $\B_{\sf h}(\C^4)$ is the set of $4\times4$
hermitian matrices, and define the perturbed Hamiltonian $H_\pm:=H_0\pm V$. Since $V$ is
bounded and symmetric, the operator $H_\pm$ is selfadjoint in $\H$ and has domain
$\dom(H)=\dom(H_0)$. We also assume that $|V(x)|$ decays more rapidly than $|x|^{-3}$
as $|x|\to\infty$ and that
\begin{equation}\label{durdur}
(H_\pm-z)^{-3}-(H_0-z)^{-3}\in S_1(\H)
\quad\hbox{for each}\quad z\in\R\setminus\{\sigma(H_0)\cup\sigma(H_\pm)\},
\end{equation}
where $S_1(\H)$ denotes the set of trace class operators in $\H$.

Under these assumptions, there exists a unique function
$\xi(\,\cdot\,;H_\pm,H_0)\in\lone\big(\R;(1+|\lambda|)^{-4}\d\lambda\big)$ such that the
Lifshits-Krein trace formula
\begin{equation}\label{eq_LK}
\tr\big[f(H_\pm)-f(H_0)\big]
=\int_\R\d\lambda\,f'(\lambda)\;\!\xi(\lambda;H_\pm,H_0)
\end{equation}
holds for each $f\in C^\infty_0(\R)$ (see \cite[Sec.~8.11]{Yaf92}). The function
$\xi(\,\cdot\,;H_\pm,H_0)$ is called the spectral shift function for the pair $(H_\pm,H_0)$.
It vanishes identically on $\R\setminus\{\sigma(H_0)\cup\sigma(H_\pm)\}$, and can be related
to the number of eigenvalues of $H_\pm$ in $(-m,m)$ (see Remark
\ref{int_eigen}). Morever, for almost every $\lambda\in\sigma_{\rm ac}(H_0)$ the spectral
shift function is related to the scattering matrix $S(\lambda;H_\pm,H_0)$ for the pair
$(H_\pm,H_0)$ by the Birman-Krein formula
$$
\det S(\lambda;H_\pm,H_0)=\e^{-2\pi i\xi(\lambda;H_\pm,H_0)}.
$$

After identification of $\xi(\,\cdot\,;H_\pm,H_0)$ with some representative of
its equivalence class, our results are the following. In Proposition \ref{prop_cont}, we
show that there exists a constant $\zeta>0$ defined in terms of $b$ (\cf Proposition
\ref{zeta}) such that $\xi(\,\cdot\,;H_\pm,H_0)$ is bounded on each compact subset of
$(-\sqrt{m^2+\zeta},\sqrt{m^2+\zeta})\setminus\{\pm m\}$ and is continuous on
$(-\sqrt{m^2+\zeta},\sqrt{m^2+\zeta})\setminus\big(\{\pm m\}\cup\sigma_{\rm p}(H_\pm)\big)$.
In Theorem \ref{thm<}, we determine the asymptotic behaviour of
$\xi(\lambda;H_\pm,H_0)$ as $\lambda\to\pm m$, $|\lambda|<m$, and in Theorem \ref{thm_ext},
we determine the asymptotic behaviour of $\xi(\lambda;H_\pm,H_0)$ as
$\lambda\to\pm m$, $|\lambda|>m$. In both cases, one has
$\xi(\lambda;H_\pm,H_0)\to\pm\infty$ as $\lambda\to\mp m$. The divergence of
$\xi(\lambda;H_\pm,H_0)$ near $\lambda=\pm m$ scales as the number of eigenvalues near $0$ of
certain Berezin-Toeplitz type operators. When $V$ admits a power-like or exponential decay at
infinity, or when it has a compact support, we give the first term of the asymptotic expansion of
$\xi(\lambda,;H_\pm,H_0)$ near $\lambda=\pm m$ (see Proposition \ref{in_gap} and Corollary
\ref{outside_gap}). In these cases, we show that the limits
$$
\lim_{\varepsilon\searrow0}
\frac{\xi\big(m+\varepsilon;H_-,H_0\big)}{\xi\big(m-\varepsilon;H_-,H_0\big)}
\qquad\hbox{and}\qquad
\lim_{\varepsilon\searrow0}
\frac{\xi\big(-m-\varepsilon;H_+,H_0\big)}{\xi\big(-m+\varepsilon;H_+,H_0\big)}
$$
exist and are equal to positive constants depending on the decay rate of $V$ at infinity
(see Corollary \ref{Levinson} for a precise statement). This can be interpreted as a
generalised version of Levinson's Theorem for the pair $(H_\pm,H_0)$ (see \cite{Kla90,Ma06}
for usual versions of Levinson's Theorem for Dirac operators). The relation between
the behaviour of the spectral shift function near $\lambda=+m$ and near $\lambda=-m$
is explained in Remark \ref{C_sym} by using the charge conjugation symmetry.

These results are similar to the results of \cite{Rai09} (where Pauli operators with
non-constant magnetic fields are considered) and \cite{FR04} (where Schr\"odinger operators
with constant magnetic field are considered). Part of the interest of this work relies on
the fact that we were able to exhibit a non-trivial class of matrix potentials $V$ satisfying \eqref{durdur} even though $H_0$ is not a bounded perturbation of the free Dirac operator.
We refer to Remark \ref{a_class} and Section \ref{cond_trace} for a discussion of this issue.

Let us fix the notations that are used in the paper. The norm and scalar product of
$\H\equiv\ltwo(\R^3;\C^4)$ are denoted by $\|\;\!\cdot\;\!\|$ and
$\langle\;\!\cdot\;\!,\;\!\cdot\;\!\rangle$. The symbol $\otimes$ stands for the closed
tensor product of Hilbert spaces and $S_p(\H)$, $p\in[1,\infty]$, denotes the $p$-th
Schatten-von Neumann class of operators in $\H$ ($S_\infty(\H)$ is the set of compact operators
in $\H$). We denote by $\|\;\!\cdot\;\!\|_p$ the corresponding operator norm. The variable
$x\in\R^3$ is often written as $x\equiv(x_\perp,x_3)$, with $x_\perp\in\R^2$ and $x_3\in\R$.
The symbol $Q_j$, $j=1,2,3$, denotes the multiplication operator by $x_j$ in $\H$,
$Q:=(Q_1,Q_2,Q_3)$, and $Q_\perp:=(Q_1,Q_2)$. Sometimes, when the context is unambiguous, we
consider the operators $Q_j$ and $P_j$ as operators in $\ltwo(\R)$ instead of $\H$ without
changing the notations. Given a selfadjoint operator $A$ in a Hilbert space $\G$, the symbol
$E^A(\;\!\cdot\;\!)$ stands for the spectral measure of $A$.

%-----------------------------------------------------------------------------------------------
\section{Unperturbed operator}\label{Unperturbed}
\setcounter{equation}{0}
%-----------------------------------------------------------------------------------------------

Throughout this paper we assume that the component $b:\R^2\to\R$ of the magnetic field
$\vec B\equiv(0,0,b)$ belongs to the class of ``admissible" magnetic fields defined in
\cite[Sec.~2.1]{Rai09}. Namely, we assume that $b=b_0+\widetilde b$, where $b_0>0$ is a
constant while the function $\widetilde b:\R^2\to\R$ is such that the Poisson equation
$$
\Delta\widetilde\varphi=\widetilde b
$$
admits a solution $\widetilde\varphi:\R^2\to\R$, continuous and
bounded together with its derivatives of order up to two. We also define
$\varphi_0(x_\perp):=\frac14b_0|x_\perp|^2$ for each $x_\perp\in\R^2$ and set
$\varphi:=\varphi_0+\widetilde\varphi$. Then we obtain a vector potential
$\vec a\equiv(a_1,a_2,a_3)\in C^1(\R^2;\R^3)$ for the magnetic field $\vec B$ by putting
$$
a_1:=\partial_1\varphi,\qquad a_2:=\partial_2\varphi\qquad\hbox{and}\qquad a_3:=0.
$$
(changing, if necessary, the gauge, we shall always assume that the vector potential
$\vec a$ is of this form). We refer to \cite{Rai09} for further properties and examples of
admissible magnetic fields.

Since the vector potential $\vec a$ belongs to
$\linfloc(\R^2;\R^3)$, the magnetic Dirac operator
$$
H_0=\alpha_1\Pi_1+\alpha_2\Pi_2+\alpha_3P_3+\beta m
$$
satisfies all the properties of \cite[Sec. 2.1]{RT07}. The operator $H_0$ is
essentially selfadjoint on $C^\infty_0(\R^3;\C^4)$, with domain
$\dom(H_0)\subset\H^{1/2}_{\rm loc}(\R^3;\C^4)$, the spectrum of $H_0$ satisfies
\begin{equation}\label{sigma_0}
\sigma(H_0)=\sigma_{\rm ac}(H_0)=(-\infty,-m]\cup[m,\infty),
\end{equation}
and we have the identity
\begin{equation}\label{tortuga}
H_0^2=
\(\begin{smallmatrix}
H_\perp^-\otimes1+1\otimes(P_3^2+m^2) & 0 & 0 & 0\\
0 & H_\perp^+\otimes1+1\otimes(P_3^2+m^2) & 0 & 0\\
0 & 0 & H_\perp^-\otimes1+1\otimes(P_3^2+m^2) & 0\\
0 & 0 & 0 & H_\perp^+\otimes1+1\otimes(P_3^2+m^2)
\end{smallmatrix}\)
\end{equation}
with respect to the tensorial decomposition $\ltwo(\R^2)\otimes\ltwo(\R)$ of $\ltwo(\R^3)$.
Here the operators $H_\perp^\pm$ are the components of the Pauli operator
$H_\perp:=H_\perp^-\oplus H_\perp^+$ in $\ltwo(\R^2;\C^2)$ associated with the vector
potential $(a_1,a_2)$. 

We recall from \cite[Sec.~2.2]{Rai09} that $\dim\ker(H_\perp^-)=\infty$, that
$\dim\ker(H_\perp^+)=0$ and that we have the following result.

\begin{Proposition}\label{zeta}
Let $b$ be an admissible magnetic field with $b_0>0$. Then $0=\inf\sigma(H_\perp)$
is an isolated eigenvalue of infinite multiplicity. More precisely, we have
$$
\dim\ker(H_\perp)=\infty\qquad\hbox{and}\qquad
(0,\zeta)\subset\R\setminus\sigma(H_\perp),
$$
where
$$
\zeta:=2b_0\e^{-2{\rm osc(\widetilde\varphi)}}\qquad\hbox{and}\qquad
{\rm osc(\widetilde\varphi)}:=\sup_{x_\perp\in\R^2}\widetilde\varphi(x_\perp)
-\inf_{x_\perp\in\R^2}\widetilde\varphi(x_\perp).
$$
\end{Proposition}

Finally, since $(0,{\zeta})\subset\R\setminus\sigma(H_\perp)$, we know from
\cite[Thm. 1.2.(d)]{RT07} that the limits
\begin{equation}\label{general_LAP}
\lim_{\varepsilon\searrow0}
\langle Q_3\rangle^{-\nu_3/2}(H_0-\lambda\mp i\varepsilon)^{-1}
\langle Q_3\rangle^{-\nu_3/2},\qquad \nu_3>1,
\end{equation}
exist for each $\lambda\in(-\sqrt{m^2+\zeta},\sqrt{m^2+\zeta})\setminus\{\pm m\}$ (note that
we use the usual notation $\langle\;\!\cdot\;\!\rangle:=\sqrt{1+|\;\!\cdot\;\!|^2}$).

%-----------------------------------------------------------------------------------------------
\section{Perturbed operator}\label{Perturbed}
\setcounter{equation}{0}
%-----------------------------------------------------------------------------------------------

We consider now the perturbed operators $H_\pm=H_0\pm V$, where $V\equiv\{V_{jk}\}$ is the
multiplication operator associated to the following matrix-valued function $V$.

\begin{Assumption}\label{assumption1}
The function $V\in C\big(\R^3;\B_{\sf h}(\C^4)\big)$ satisfies for each
$x\equiv(x_\perp,x_3)\in\R^3$ and each $j,k\in\{1,\ldots,4\}$
\begin{equation}\label{first_decay}
V(x)\ge0\qquad\hbox{and}\qquad
|V_{jk}(x)|\le{\rm Const.}\;\!\langle x_\perp\rangle^{-\nu_\perp}\langle x_3\rangle^{-\nu_3}
\quad\hbox{for some }\nu_\perp>2\hbox{ and }\nu_3>1.
\end{equation}
\end{Assumption}

The potential $V$ in Assumption \ref{assumption1} is short-range along $x_3$. So
we know from \cite[Thm. 1.2]{RT07} that
\begin{enumerate}
\item[(i)] $\sigma_{\rm ess}(H_\pm)=\sigma_{\rm ess}(H_0)=(-\infty,-m]\cup[m,\infty)$.
\item[(ii)] The point spectrum of $H_\pm$ in
$
\big(-\sqrt{m^2+\zeta},\sqrt{m^2+\zeta}\big)\setminus\{\pm m\}
$
is composed of eigenvalues of finite multiplicity and with no accumulation point.
\item[(iii)] $H_\pm$ has no singular continuous spectrum in
$
\big(-\sqrt{m^2+\zeta},\sqrt{m^2+\zeta}\big)
$.
In particular, $H_0$ and $H_\pm$ have a common spectral gap in $(-m,m)$.
\end{enumerate}

Using the formula
$$
(A+\lambda)^{-\gamma}
=\Gamma(\gamma)^{-1}\int_0^\infty\d t\,t^{\gamma-1}\e^{-t(A+\lambda)},
\qquad A:\dom(A)\to\H,~A\ge0,~\lambda,\gamma>0,
$$
the diamagnetic inequality \cite[Thm. 2.3]{AHS78}, and the compactness criterion
\cite[Thm. 5.7.1]{Dav07}, we find that
$$
\textstyle
|V_{jk}|^{1/2}\big(\sum_{\ell\le3}\Pi_\ell^*\Pi_\ell+m^2\big)^{-1/4}
\in S_\infty[\ltwo(\R^3)].
$$
Since $b$ is bounded this implies that
$$
|H_0|^{-1/2}V|H_0|^{-1/2}
\le\textstyle|H_0|^{-1/2}\big(\sum_{j,k\le4}|V_{jk}|\big)|H_0|^{-1/2}
\in S_\infty(\H).
$$
So $|H_0|^{-1/2}V|H_0|^{-1/2}$ also belongs to $S_\infty(\H)$, since $S_\infty(\H)$
is an hereditary $C^*$-subalgebra of $\B(\H)$ \cite[Cor. 3.2.3]{Mur90}. One has in
particular
\begin{equation}\label{necessary1}
V^{1/2}(|H_0|+1)^{-1/2}\in S_\infty(\H).
\end{equation}
The standard criterion \cite[Thm. XI.20]{RS79} shows that
$$
|V_{jk}|^{1/2}\big(-\Delta+m^2\big)^{-\gamma}\in S_q[\ltwo(\R^3)]
\quad\hbox{if }q\in[2,\infty)\hbox{ and }\gamma q>3/2.
$$
This together with arguments as above
implies that
\begin{equation}\label{necessary2}
V^{1/2}|H_0|^{-\gamma}\in S_q(\H)
\quad\hbox{if }q\ge2\hbox{ is even and }\gamma q>3.
\end{equation}
So we have in particular that
\begin{equation}\label{necessary3}
V^{1/2}E^{H_0}(B)\in S_2(\H)\quad\hbox{for any bounded borel set }B\subset\R.
\end{equation}

In the sequel we shall need a more restrictive assumption on $V$. For this, we recall
that there exists numbers $z\in\R\setminus\{\sigma(H_0)\cup\sigma(H_\pm)\}$ since
$H_0$ and $H_\pm$ have a common spectral gap in $(-m,m)$. We also set $R_0(z):=(H_0-z)^{-1}$
and $R_\pm(z):=(H_\pm-z)^{-1}$ for $z\in\C\setminus\sigma(H_0)$ and
$z\in\C\setminus\sigma(H_\pm)$, respectively.

\begin{Assumption}\label{assumption2}
The function $V\in C\big(\R^3;\B_{\sf h}(\C^4)\big)$ satisfies for each $x\in\R^3$ and
each $j,k\in\{1,\ldots,4\}$
\begin{equation}\label{a_decay}
V(x)\ge0\qquad\hbox{and}\qquad
|V_{jk}(x)|\le{\rm Const.}\;\!\langle x\rangle^{-\nu}\quad\hbox{for some constant }\nu>3.
\end{equation}
Furthermore, $V$ is chosen such that
\begin{equation}\label{necessary4}
R_\pm^3(z)-R_0^3(z)\in S_1(\H)
\quad\emph{for each }z\in\R\setminus\{\sigma(H_0)\cup\sigma(H_\pm)\}.
\end{equation}
\end{Assumption}

Note that \eqref{a_decay} implies \eqref{first_decay} if one takes $\nu_3\in(1,\nu-2)$
and $\nu_\perp:=\nu-\nu_3$. Note also that the choice of function
$\lambda\mapsto(\lambda-z)^{-3}$ in the trace class condition \eqref{necessary4} has
been made for convenience. Many other choices would also guarantee the existence of
the spectral shift function for the pair $(H_\pm,H_0)$ (see \eg \cite[Sec.~8.11]{Yaf92}).

\begin{Remark}\label{a_class}
Since the operator $H_0$ is not a bounded perturbation of the free Dirac operator, we
cannot apply the results of \cite[Sec.~4]{Yaf05} to prove the inclusion \eqref{necessary4}
under the condition \eqref{a_decay}. In general, one has to impose additional assumptions
on $V$ to get the result. For instance, if $V$ verifies
\eqref{a_decay}, and
\begin{enumerate}
\item[(i)] $[V,\alpha_1]=[V,\alpha_2]=0$,
\item[(ii)] for each $x\in\R^3$ and each $j,k,\ell\in\{1,\ldots,4\}$, one has 
$
|(\partial_\ell V_{jk})(x)|\le{\rm Const.}\;\!\langle x\rangle^{-\varsigma}
$
for some $\varsigma>3$,
\item[(iii)] for each $j,k,\ell\in\{1,\ldots,4\}$, one has
$(\partial_\ell\partial_3V_{jk})\in\linf(\R^3)$,
\end{enumerate}
then \eqref{necessary4} is satisfied. Furthermore, if $V$ is scalar, then the same
is true without assuming (iii) (and (i) is trivially satisfied). The proof of these
statements can be found in the appendix. Here, we only note that a matrix
${\sf V}\in\B_{\sf h}(\C^4)$ satisfying (i) is necessarily of the form
$$
{\sf V}=\left(\begin{smallmatrix}
{\sf v}_1 & 0 & {\sf v}_3 & 0\\
0 & {\sf v}_2 & 0 & \overline{\,{\sf v}_3}\\
\overline{\,{\sf v}_3} & 0 & {\sf v}_2 & 0\\
0 & {\sf v}_3 & 0 & {\sf v}_1
\end{smallmatrix}\right),
$$
with ${\sf v}_1,{\sf v}_2\in\R$ and ${\sf v}_3\in\C$.
\end{Remark}

%-----------------------------------------------------------------------------------------------
\section{Spectral shift function}\label{SSF}
\setcounter{equation}{0}
%-----------------------------------------------------------------------------------------------

In this section we recall some results due to A. Pushnitski on the representation of the
spectral shift function for a pair of not semibounded selfadjoint operators.

Given a a Lebesgue measurable set $B\subset\R$, we set
$\mu(B):=\frac1\pi\int_B\frac{\d t}{1+t^2}$, and note that $\mu(\R)=1$. Furthermore, if $T=T^*$
is a compact operator in a separable Hilbert space $\G$, we set
$$
n_\pm(s;T):=\rank E^{\pm T}\big((s,\infty)\big)\quad\hbox{for }s>0.
$$
Then we have the following estimates.

\begin{Lemma}[Lemma 2.1 of \cite{Pus97}]\label{aday}
Let $T_1=T_1^*\in S_\infty(\H)$ and $T_2=T_2^*\in S_1(\H)$. Then one as for each $s_1,s_2>0$
$$
\int_\R\d\mu(t)\,n_\pm(s_1+s_2;T_1+tT_2)
\leq n_{\pm}(s_1;T_1)+\frac1{\pi s_2}\;\!\|T_2\|_1.
$$
\end{Lemma}

For $z\in\C\setminus\sigma(H_0)$, we define the usual weighted resolvent
$$
T(z):=V^{1/2}(H_0-z)^{-1}V^{1/2}
$$
and the corresponding real and imaginary parts
$$
A(z):=\re T(z)\qquad\hbox{and}\qquad B(z):=\im T(z).
$$
Then the next lemma is direct consequence of the inclusions
\eqref{necessary1}-\eqref{necessary3} and \cite[Prop.~4.4.(i)]{Pus01}.

\begin{Lemma}\label{2limits}
Let $V$ satisfy Assumption \ref{assumption1}. Then, for almost every $\lambda\in\R$, the
limits $A(\lambda+i0):=\lim_{\varepsilon\searrow0}A(\lambda+i\varepsilon)$ and
$B(\lambda+i0):=\lim_{\varepsilon\searrow0}B(\lambda+i\varepsilon)\ge0$ exist in
$S_4(\H)$.
\end{Lemma}

Next theorem follows from the inclusions \eqref{necessary1}, \eqref{necessary3},
\eqref{necessary4}, from the equations (1.9), (8.1), (8.2) of \cite{Pus01}, and from
Theorem 8.1 of \cite{Pus01}.

\begin{Theorem}\label{identify}
Let $V$ satisfy Assumption \ref{assumption2}. Then, for almost every $\lambda\in\R$,
$\xi(\lambda;H_\pm,H_0)$ exists and is given by
\begin{equation}\label{lhs}
\xi(\lambda;H_\pm,H_0)
=\pm\int_\R\d\mu(t)\,n_\mp\big(1;A(\lambda+i0)+tB(\lambda+i0)\big).
\end{equation}
\end{Theorem}

We know from \eqref{general_LAP} that $A(\lambda+i0)$ and $B(\lambda+i0)$ exist in $\B(\H)$
for each
$
\lambda\in(-\sqrt{m^2+\zeta},\sqrt{m^2+\zeta})\setminus\{\pm m\}
$.
In Propositions \ref{coffee}-\ref{cigarette} and Corollary \ref{inParis} below we show
that in fact $A(\lambda+i0)\in S_4(\H)$ and $B(\lambda+i0)\in S_1(\H)$ for each
$
\lambda\in(-\sqrt{m^2+\zeta},\sqrt{m^2+\zeta})\setminus\{\pm m\}
$.
Hence, by Lemma \ref{aday}, the r.h.s. of \eqref{lhs} will turn out to be well-defined for
every
$
\lambda\in(-\sqrt{m^2+\zeta},\sqrt{m^2+\zeta})\setminus\{\pm m\}
$.
In the next proposition we state some regularity properties of the function
$$
(-\sqrt{m^2+\zeta},\sqrt{m^2+\zeta})\setminus\{\pm m\}\ni\lambda\mapsto
\widetilde\xi(\lambda;H_\pm,H_0)
:=\pm\int_\R\d\mu(t)\,n_\mp\big(1;A(\lambda+i0)+tB(\lambda+i0)\big).
$$
The proof (which relies on Propositions \ref{coffee}-\ref{cigarette}, Lemma \ref{rank2},
Corollary \ref{inParis} and the stability result \cite[Thm.~3.12]{GM00}) is similar to the
one of \cite[Sec. 4.2.1]{BPR04}.

\begin{Proposition}\label{prop_cont}
Let $V$ satisfy Assumption \ref{assumption1}. Then $\widetilde\xi(\,\cdot\,;H_\pm,H_0)$
is bounded on each compact subset of
$
(-\sqrt{m^2+\zeta},\sqrt{m^2+\zeta})\setminus\{\pm m\}
$
and is continuous on
$
(-\sqrt{m^2+\zeta},\sqrt{m^2+\zeta})
\setminus\big(\{\pm m\}\cup\sigma_{\rm p}(H_\pm)\big).
$
\end{Proposition}

In the sequel, we identify the functions $\widetilde\xi(\,\cdot\,;H_\pm,H_0)$ and
$\xi(\,\cdot\,;H_\pm,H_0)$ since they are equal for almost every $\lambda\in\R$ due to
Theorem \ref{identify} (see \cite{Saf01} for a study where the r.h.s. of \eqref{lhs} is
directly treated as a definition of $\xi(\lambda;H_\pm,H_0)$).

\begin{Remark}\label{int_eigen}
In the interval $(-m,m)$, $H_0$ has no spectrum and the spectrum of $H_\pm$ is purely
discrete. Thus the spectral shift function $\xi(\,\cdot\,;H_\pm,H_0)$ can be related
to the number of eigenvalues of $H_\pm$ as follows: for
$\lambda_1,\lambda_2\in(-m,m)\setminus\sigma(H_\pm)$ with $\lambda_1<\lambda_2$, we have
(see \cite[Thm.~9.1]{Pus01})
$$
\xi(\lambda_1;H_\pm,H_0)-\xi(\lambda_2;H_\pm,H_0)
=\rank E^{H_\pm}\big([\lambda_1,\lambda_2)\big).
$$
\end{Remark}

%-----------------------------------------------------------------------------------------------
\section{Decomposition of the weighted resolvent}\label{Sec_Dec}
\setcounter{equation}{0}
%-----------------------------------------------------------------------------------------------

In this section we decompose the weighted resolvent
$$
T(z)=V^{1/2}(H_0-z)V^{1/2},\quad z\in\C\setminus\sigma(H_0),
$$
into a sum $T(z)=T_{\sf div}(z)+T_{\sf bound}(z)$, where $T_{\sf div}(z)$ (respectively
$T_{\sf bound}(z)$) corresponds to the diverging (respectively non-diverging) part of $T(z)$
as $z\to\pm m$. Then we estimate the behaviour, in suitable Schatten norms, of each term as
$z\to\pm m$. We refer to \cite[Sec.~4]{FR04} and \cite[Sec.~4.2]{Rai09} for similar
approaches in the case of the Schr\"odinger and Pauli operators.

Let $\a$ and $\a^*$ be the closures in $\ltwo(\R^2)$ of the operators given by
$$
\a\varphi:=(\Pi_1-i\Pi_2)\varphi\qquad{\rm and}\qquad
\a^*\varphi:=(\Pi_1+i\Pi_2)\varphi,
$$
for $\varphi\in C^\infty_0(\R^2)$. Then one has (see \cite[Sec.~5.5.2]{Tha92} and
\cite[Sec.~5]{Rai99})
\begin{equation}\label{tortugo}
H_0=
\(\begin{smallmatrix}
m & 0 & 1\otimes P_3 & \a\otimes1\\
0 & m & \a^*\otimes1 & -1\otimes P_3\\
1\otimes P_3 & \a\otimes1 & -m & 0\\
\a^*\otimes1 & -1\otimes P_3 & 0 & -m
\end{smallmatrix}\),
\end{equation}
with
\begin{equation}\label{excel}
\ker(\a^*)=\ker(\a\a^*)=\ker(H_\perp^-)\subset\ltwo(\R^2).
\end{equation}
Now, let
$$
\P:=\(\begin{smallmatrix}
P & 0 & 0 & 0\\
0 & 0 & 0 & 0\\
0 & 0 & P & 0\\
0 & 0 & 0 & 0
\end{smallmatrix}\)
$$
be the orthogonal projection onto the union of the eigenspaces of $H_0$ corresponding to
the values $\lambda=\pm m$. Since $P\equiv p\otimes1$ is the orthogonal projection
onto $\ker(H_\perp^-)\otimes\ltwo(\R)$, the equations \eqref{tortugo} and \eqref{excel}
imply that $H_0$ and $\P$ commute:
\begin{equation}\label{acommutator}
H_0^{-1}\P=\P H_0^{-1}.
\end{equation}
In fact, by using \eqref{tortuga} and \eqref{tortugo}, one gets for each
$z\in\C\setminus\sigma(H_0)$ the equalities
\begin{align*}
&(H_0-z)^{-1}\P\\
&=(H_0+z)\big(H_0^2-z^2\big)^{-1}\P\\
&=\big[p\otimes R(z^2-m^2)\big]
\(\begin{smallmatrix}
(z+m) & 0 & 0 & 0\\
0 & 0 & 0 & 0\\
0 & 0 & (z-m) & 0\\
0 & 0 & 0 & 0
\end{smallmatrix}\)
+\big[p\otimes P_3R(z^2-m^2)\big]
\(\begin{smallmatrix}
0 & 0 & 1 & 0\\
0 & 0 & 0 & 0\\
1 & 0 & 0 & 0\\
0 & 0 & 0 & 0
\end{smallmatrix}\),
\end{align*}
where $R(z):=\big(P_3^2-z\big)^{-1}$, $z\in\C\setminus[0,\infty)$, is the resolvent of $P_3^2$
in $\ltwo(\R)$. This allows us to decompose $T(z)$ as $T(z)=T_{\sf div}(z)+T_{\sf bound}(z)$,
with
\begin{align*}
T_{\sf div}(z)&:=V^{1/2}\big[p\otimes R(z^2-m^2)\big]
\(\begin{smallmatrix}
(z+m) & 0 & 0 & 0\\
0 & 0 & 0 & 0\\
0 & 0 & (z-m) & 0\\
0 & 0 & 0 & 0
\end{smallmatrix}\)
V^{1/2},\\
T_{\sf bound}(z)&:=V^{1/2}\big[p\otimes P_3R(z^2-m^2)\big]
\(\begin{smallmatrix}
0 & 0 & 1 & 0\\
0 & 0 & 0 & 0\\
1 & 0 & 0 & 0\\
0 & 0 & 0 & 0
\end{smallmatrix}\)
V^{1/2}
+V^{1/2}(H_0-z)^{-1}\P^\perp V^{1/2}
\qquad(\P^\perp:=1-\P).
\end{align*}
One may note that this decomposition of $T(z)$ differs slightly from the simpler
decomposition
$$
T(z)=V^{1/2}(H_0-z)\P V^{1/2}+V^{1/2}(H_0-z)\P^\perp V^{1/2},
$$
since the first term in $T_{\sf bound}(z)$ is associated to the projection $\P$ and not
the projection $\P^\perp$. This choice is motivated by the will of distinguishing clearly
the contribution $T_{\sf div}(z)$, that diverge as $z\to\pm m$, from the contribution
$T_{\sf bound}(z)$, that stays bounded as $z\to\pm m$.

For $\lambda\in\R\setminus\{0\}$, we can define the boundary value $R(\lambda)$ of the
resolvent $R(z)$ as the operator with convolution kernel $r_\lambda(\,\cdot\,)$, where
$$
r_\lambda(x_3)
:=\begin{cases}
\frac{\e^{-\sqrt{-\lambda}|x_3|}}{2\sqrt{-\lambda}} & \hbox{if }\lambda<0,\vspace{3pt}\\
\frac{i\e^{\sqrt{\lambda}|x_3|}}{2\sqrt{\lambda}} & \hbox{if }\lambda>0,
\end{cases}
$$
for each $x_3\in\R$. So, we can extend the definition of $T_{\sf div}(\,\cdot\,)$ to the
values $\lambda\in\R\setminus\{\pm m\}$:
$$
T_{\sf div}(\lambda):=V^{1/2}\big[p\otimes R(\lambda^2-m^2)\big]
\(\begin{smallmatrix}
(\lambda+m) & 0 & 0 & 0\\
0 & 0 & 0 & 0\\
0 & 0 & (\lambda-m) & 0\\
0 & 0 & 0 & 0
\end{smallmatrix}\)
V^{1/2}.
$$

In the following proposition, we show that the trace norm of $T_{\sf div}(z)$ is continuous
in $\C_+:=\{z\in\C\mid\im(z)\ge0\}$ outside the points $z=\pm m$, where it may
diverge as $|z\mp m|^{-1/2}$. The proof of the proposition relies on a technical result
that we now recall.

\begin{Lemma}[Lemma 2.4 of \cite{Rai09}]\label{lem_rai}
Let $U\in\mathsf{L}^{\:\!\!q}(\R^2)$, $q\in[1,\infty)$, and assume that $b$ is an admissible
magnetic field. Then $pUp\in S_q[\ltwo(\R^2)]$, and
$$
\big\|pUp\big\|_{S_q[\ltwo(\R^2)]}^q
\le\frac{b_0}{2\pi}\;\!\e^{2{\rm osc(\widetilde\varphi)}}\|U\|_{\mathsf{L}^{\:\!\!q}(\R^2)}^q\,.
$$
\end{Lemma}

The symbol $y_+$ denotes the postive part of $y\in\R$.

\begin{Proposition}\label{coffee}
Let $V$ satisfy Assumption \ref{assumption1}. Then the operator-valued function
$$
\C_+\setminus\{\pm m\}\ni z\mapsto T_{\sf div}(z)\in S_1(\H)
$$
is well-defined and continuous. Moreover, we have for each
$\lambda\in\R\setminus\{\pm m\}$ the bound
$$
\|T_{\sf div}(\lambda)\|_1
\le{\rm Const.}\;\!\textstyle\Big(\big|\frac{\lambda+m}{\lambda-m}\big|^{1/2}
+\big|\frac{\lambda-m}{\lambda+m}\big|^{1/2}\Big)\big(1+(\lambda^2-m^2)_+^{1/4}\big).
$$
\end{Proposition}

\begin{proof}
We have for each $z\in\C\setminus\sigma(H_0)$ the identity
$$
T_{\sf div}(z)=M\big(G\otimes J_{z^2-m^2}\big)
\(\begin{smallmatrix}
(z+m) & 0 & 0 & 0\\
0 & 0 & 0 & 0\\
0 & 0 & (z-m) & 0\\
0 & 0 & 0 & 0
\end{smallmatrix}\)M,
$$
where
\begin{align}
M&:=V^{1/2}\langle Q_\perp\rangle^{\nu_\perp/2}\langle Q_3\rangle^{\nu_3/2},
\label{operatorM}\\
G&:=\langle Q_\perp\rangle^{-\nu_\perp/2}p\langle Q_\perp\rangle^{-\nu_\perp/2},
\label{operatorG}\\	
J_z&:=\langle Q_3\rangle^{-\nu_3/2}R(z)\langle Q_3\rangle^{-\nu_3/2}.\nonumber
\end{align}
The operator $M$ is bounded due to Assumption \ref{assumption1}. So
$$
\|T_{\sf div}(z)\|_1\le{\rm Const.}\(|z+m|+|z-m|\)\|G\|_1\|J_{z^2-m^2}\|_1.
$$
But we know from Lemma \ref{lem_rai} that $\|G\|_1\le{\rm Const.}$, and from
\cite[Sec.~4.1]{BPR04} that the operator-valued function
$\C_+\setminus\{0\}\ni z\mapsto J_z$ is continuous in the trace norm and admits the bound
$$
\|J_\lambda\|_1\le{\rm Const.}\;\!\big(1+\lambda_+^{1/4}\big)|\lambda|^{-1/2},
\quad\lambda\in\R\setminus\{0\}.
$$
It follows that
$$
\|T_{\sf div}(z)\|_1\le{\rm Const.}\;\!\textstyle
\Big(\big|\frac{\lambda+m}{\lambda-m}\big|^{1/2}
+\big|\frac{\lambda-m}{\lambda+m}\big|^{1/2}\Big)
\big(1+(\lambda^2-m^2)_+^{1/4}\big)
$$
for each $\lambda\in\R\setminus\{\pm m\}$.
\end{proof}

In the following proposition, we show that the function
$z\mapsto T_{\sf bound}(z)\in S_4(\H)$ is continuous in
$
\C\setminus\big\{(-\infty,-\sqrt{m^2+\zeta}]\cup[\sqrt{m^2+\zeta},\infty)\big\}.
$
The symbols $H^\pm$ stand for the operators $H^\pm:=H_\perp^\pm\otimes1+1\otimes P_3^2$
acting in $\ltwo(\R^3)$. 

\begin{Proposition}\label{cigarette}
Let $V$ satisfy Assumption \ref{assumption1}. Then the operator-valued function
$$
\C\setminus\big\{(-\infty,-\sqrt{m^2+\zeta}]\cup[\sqrt{m^2+\zeta},\infty)\big\}
\ni z\mapsto T_{\sf bound}(z)\in S_4(\H)
$$
is well-defined and continuous. Moreover, we have for each
$
\lambda\in(-\sqrt{m^2+\zeta},\sqrt{m^2+\zeta})
$
the bound
\begin{equation}\label{ChuchoValdes}
\textstyle\|T_{\sf bound}(\lambda)\|_4\le{\rm Const.}\(|\lambda|+\lambda^2\)
\Big(1+\frac{(\lambda^2-m^2+1)_+}{\zeta+m^2-\lambda^2}\Big)+{\rm Const.}
\end{equation}
\end{Proposition}

\begin{proof}
One has the identity
$$
(H_0-z)^{-1}=H_0^{-1}+z\big(1+zH_0^{-1}\big)\big(H_0^2-z^2\big)^{-1}
$$
for each $z\in\C\setminus\sigma(H_0)$. Thus the operator $T_{\sf bound}(z)$ can be
written as
\begin{align}
T_{\sf bound}(z)&=M\big(G\otimes S_z\big)
\(\begin{smallmatrix}
0 & 0 & 1 & 0\\
0 & 0 & 0 & 0\\
1 & 0 & 0 & 0\\
0 & 0 & 0 & 0
\end{smallmatrix}\)M
+V^{1/2}H_0^{-1}\P^\perp V^{1/2}
+zV^{1/2}\big(1+zH_0^{-1}\big)\big(H_0^2-z^2\big)^{-1}
\P^\perp V^{1/2}\label{BeforeCommute}\\
&\equiv T_1(z)+T_2+T_3(z),\nonumber
\end{align}
with $M$ and $G$ given by \eqref{operatorM}-\eqref{operatorG}, and
$$
S_z:=\langle Q_3\rangle^{-\nu_3/2}P_3R(z^2-m^2)\langle Q_3\rangle^{-\nu_3/2}.
$$
The integral kernel of $S_z$ is
\begin{equation}\label{truffes}
\textstyle\frac i2\langle x_3\rangle^{-\nu_3/2}
\frac{(x_3-x_3')}{|x_3-x_3'|}\e^{i\sqrt{z^2-m^2}|x_3-x_3'|}
\langle x_3'\rangle^{-\nu_3/2},
\end{equation}
with the branch of $\sqrt{z^2-m^2}$ chosen so that $\im\sqrt{z^2-m^2}>0$. So $S_z$
extends to an element of $S_2[\ltwo(\R)]$ for each $z\in\C$, with $\|S_z\|_2\le{\rm Const}$.
Since $M$ is bounded and $\|G\|_1\le{\rm Const.}$, this implies that
\begin{equation}\label{firstterm}
\|T_1(z)\|_2
\le{\rm Const.}\;\!\|M\|^2\|G\|_1\|S_z\|_2
\le{\rm Const.}
\end{equation}
for each $z\in\C$. One also has
\begin{equation}\label{secondterm}
\|T_2\|_4\le{\rm Const.}
\end{equation}
due to \eqref{necessary2}. So, it only remains to bound the term $T_3(z)$.

Let
$
z\in\C\setminus\{(-\infty,-\sqrt{m^2+\zeta}]\cup[\sqrt{m^2+\zeta},\infty)\}
$
and $P^\perp:=1-P$. Then $\big(H^-+m^2-z^2\big)^{-1}P^\perp$ and
$\big(H^++m^2-z^2\big)^{-1}$ belong to $\B[\ltwo(\R^3)]$, and we have
$$
\big(H^-+m^2-z^2\big)^{-1}P^\perp=P^\perp\big(H^-+m^2-z^2\big)^{-1}.
$$
Thus
\begin{align*}
\big(H_0^2-z^2\big)^{-1}\P^\perp V^{1/2}
&=\big(H_0^2-z^2\big)^{-1}
\(\begin{smallmatrix}
P^\perp & 0 & 0 & 0\\
0 & 1 & 0 & 0\\
0 & 0 & P^\perp & 0\\
0 & 0 & 0 & 1
\end{smallmatrix}\)
V^{1/2}\\
&=\(\begin{smallmatrix}
P^\perp(H^-+m^2-z^2)^{-1} & 0 & 0 & 0\\
0 & (H^++m^2-z^2)^{-1} & 0 & 0\\
0 & 0 & P^\perp(H^-+m^2-z^2)^{-1} & 0\\
0 & 0 & 0 & (H^++m^2-z^2)^{-1}
\end{smallmatrix}\)
V^{1/2},
\end{align*}
and
$$
\big\|\big(H_0^2-z^2\big)^{-1}\P^\perp V^{1/2}\big\|^2_2
\le2\;\!\|M\|^2\Big\{\big\|P^\perp\big(H^-+m^2-z^2\big)^{-1}M_2\big\|^2_2
+\big\|\big(H^++m^2-z^2\big)^{-1}M_2\big\|^2_2\Big\},
$$
where $M_2:=\langle Q_\perp\rangle^{-\nu_\perp/2}\langle Q_3\rangle^{-\nu_3/2}$.
But, we know from the proof of \cite[Prop. 4.4]{Rai09} that
$$
\big\|P^\perp(H^-+m^2-z^2)^{-1}M_2\big\|_2\le{\rm Const.}\;\!C(z)
\qquad\hbox{and}\qquad
\big\|(H^++m^2-z^2)^{-1}M_2\big\|_2\le{\rm Const.}\;\!C(z),
$$
where
$$
C(z):=\sup_{y\in[\zeta,\infty)}\frac{y+1}{|y+m^2-z^2|}\,.
$$
It follows that
\begin{equation}\label{thirdterm}
\|T_3(z)\|_2
\le{\rm Const.}\big\|zV^{1/2}\big(1+zH_0^{-1}\big)\big\|\,\|M\|\,C(z)
\le{\rm Const.}\(|z|+|z|^2\)C(z).
\end{equation}
The claim follows then by putting together \eqref{firstterm}, \eqref{secondterm}, and
\eqref{thirdterm}.
\end{proof}

In the next lemma we give some results on the imaginary part of the operator $S_z$
in $\ltwo(\R)$ appearing in the proof of Proposition \ref{cigarette}
$$
S_z=\langle Q_3\rangle^{-\nu_3/2}P_3R(z^2-m^2)\langle Q_3\rangle^{-\nu_3/2},
\qquad z\in\C\setminus\sigma(H_0),~\nu_3>1.
$$

\begin{Lemma}\label{rank2}
\begin{enumerate}
\item[(a)] One has $\im S_\lambda=0$ for each $\lambda\in(-m,m)$.
\item[(b)] Let $p\ge1$ be an integer. Then one has for each $\lambda\in\R$ with
$|\lambda|>m$
$$
\|\im S_\lambda\|_p\le\textsc c_p\,,
$$
where $\textsc c_p$ is a constant independent of $\lambda$. Furthermore
$$
\lim_{\lambda\to\pm m,\,|\lambda|>m}\|\im S_\lambda\|_p=0.
$$
\end{enumerate}
\end{Lemma}

\begin{proof}
(a) This is a direct consequence of the spectral theorem.

(b) Let $\lambda\in\R$, $|\lambda|>m$. Then one shows by using \eqref{truffes} that
$\im S_\lambda$ is equal to the rank two operator
$$
\im S_\lambda
=\langle v_\lambda,\;\!\cdot\;\!\rangle\,u_\lambda
+\langle u_\lambda,\;\!\cdot\;\!\rangle\,v_\lambda,
$$
with
$$
\textstyle
u_\lambda(x_3):=\langle x_3\rangle^{-\nu_3/2}\sin\big(x_3\sqrt{\lambda^2-m^2}\big)
\qquad\hbox{and}\qquad
v_\lambda(x_3):=\textstyle
-\frac i2\langle x_3\rangle^{-\nu_3/2}\cos\big(x_3\sqrt{\lambda^2-m^2}\big).
$$
Since $\langle v_\lambda,u_\lambda\rangle=0$, this implies that
$$
|\im S_\lambda|^p
=\|u_\lambda\|^p\langle v_\lambda,\;\!\cdot\;\!\rangle\,v_\lambda
+\|v_\lambda\|^p\langle u_\lambda,\;\!\cdot\;\!\rangle\,u_\lambda.
$$
Thus
$$
\|\im S_\lambda\|_p^p
=\tr\big(|\im S_\lambda|^p\big)
=\|u_\lambda\|^p\,\|v_\lambda\|^2+\|v_\lambda\|^p\,\|u_\lambda\|^2.
$$
This, together with the equality
$$
\lim_{\lambda\to\pm m,\,|\lambda|>m}\|u_\lambda\|=0,
$$
implies the claim.
\end{proof}

In the next corollary we combine some of the results of Propositions \ref{coffee},
\ref{cigarette} and Lemma \ref{rank2}.

\begin{Corollary}\label{inParis}
Let $V$ satisfy Assumption \ref{assumption1}. Then the identity
\begin{equation}\label{manjar}
T(\lambda+i0)=T_{\sf div}(\lambda)+T_{\sf bound}(\lambda)
\end{equation}
holds for each
$
\lambda\in(-\sqrt{m^2+\zeta},\sqrt{m^2+\zeta})\setminus\{\pm m\}
$,
and the estimate
\begin{equation}\label{queso}
\big\|\im T_{\sf bound}(\lambda)\big\|_p\le{\rm Const.}\;\!\|\im S_\lambda\|_p
\end{equation}
holds for each integer $p\ge1$ an each
$
\lambda\in(-\sqrt{m^2+\zeta},\sqrt{m^2+\zeta}).
$
In particular, we have
\begin{equation}\label{limit_im}
\lim_{\lambda\to\pm m}\big\|\im T_{\sf bound}(\lambda)\big\|_p=0,
\end{equation}
due to Lemma \ref{rank2}.
\end{Corollary}

\begin{proof}
The first identity follows from Propositions \ref{coffee} and \ref{cigarette}. Let
$\lambda\in(-\sqrt{m^2+\zeta},\sqrt{m^2+\zeta})$. Using \eqref{BeforeCommute} and the
commutation rule \eqref{acommutator} one obtains that
$$
\im T_{\sf bound}(\lambda)=M\big(G\otimes\im S_\lambda\big)
\(\begin{smallmatrix}
0 & 0 & 1 & 0\\
0 & 0 & 0 & 0\\
1 & 0 & 0 & 0\\
0 & 0 & 0 & 0
\end{smallmatrix}\)M,
$$
with $M$ and $G$ defined by \eqref{operatorM}-\eqref{operatorG}. Since
$M$ is bounded and $\|G\|_1\le{\rm Const.}$, this implies \eqref{queso}.
\end{proof}

%-----------------------------------------------------------------------------------------------
\section{Proof of the main results}
\setcounter{equation}{0}
%-----------------------------------------------------------------------------------------------

We begin this section by showing that the value of $\xi(\lambda;H,H_\pm)$ as $\lambda\to\pm m$
is bounded from below and from above by expressions involving only the term
$T_{\sf div}(\lambda)$ of the decomposition
$T(\lambda+i0)=T_{\sf div}(\lambda)+T_{\sf bound}(\lambda)$. Then we consider separately
the limits $\lambda\to\pm m$ with $|\lambda|<m$ and the limits $\lambda\to\pm m$ with
$|\lambda|>m$.

We start by recalling two standard properties of the counting functions $n_\pm$. Given two
compact operators $T_1=T_1^*$ and $T_2=T_2^*$ in a separable Hilbert space $\G$, we have the
Weyl inequalities
\begin{equation}\label{Weyl}
n_\pm(s_1+s_2;T_1+T_2)\le n_\pm(s_1;T_1)+n_\pm(s_2;T_2)\quad\hbox{for each }s_1,s_2>0.
\end{equation}
Moreover, if $T=T^*$ belongs to $S_p(\G)$ for some $p\in[1,\infty)$, then
\begin{equation}\label{pbound}
n_\pm(s;T)\le s^{-p}\|T\|^p_p\quad\hbox{for each }s>0.
\end{equation}

\begin{Proposition}\label{empanada}
Let $V$ satisfy Assumption \ref{assumption2}. Then the estimates
\begin{align*}
&\int_\R\d\mu(t)\,
n_\pm\big(1+\varepsilon;\re T_{\sf div}(\lambda)+t\im T_{\sf div}(\lambda)\big)+\O(1)\\
&\le\mp\xi(\lambda;H_\mp,H_0)\\
&\le\int_\R\d\mu(t)\,
n_\pm\big(1-\varepsilon;\re T_{\sf div}(\lambda)+t\im T_{\sf div}(\lambda)\big)+\O(1)
\end{align*}
hold as $\lambda\to\pm m$ for each $\varepsilon\in(0,1)$.
\end{Proposition}

\begin{proof}
Using \eqref{manjar}, the Weyl inequalities \eqref{Weyl}, and Lemma \ref{aday}
we get
\begin{align}
&\int_\R\d\mu(t)\,
n_\pm\big(1+\varepsilon;\re T_{\sf div}(\lambda)+t\im T_{\sf div}(\lambda)\big)
-n_\mp\big(\varepsilon/2;\re T_{\sf bound}(\lambda)\big)
-\frac2{\pi\varepsilon}\big\|\im T_{\sf bound}(\lambda)\big\|_1\nonumber\\
&\le\int_\R\d\mu(t)\,n_\pm\big(1;A(\lambda+i0)+tB(\lambda+i0)\big)\nonumber\\
&\le\int_\R\d\mu(t)\,
n_\pm\big(1-\varepsilon;\re T_{\sf div}(\lambda)+t\im T_{\sf div}(\lambda)\big)
+n_\pm\big(\varepsilon/2;\re T_{\sf bound}(\lambda)\big)
+\frac2{\pi\varepsilon}\big\|\im T_{\sf bound}(\lambda)\big\|_1.\label{gniarf}
\end{align}
Due to \eqref{pbound}, we have
$$
n_{\pm}\big(\varepsilon/2;\re T_{\sf bound}(\lambda)\big)
\le16\!\;\varepsilon^{-4}\|T_{\sf bound}(\lambda)\|_4^4,
$$
which combined with \eqref{ChuchoValdes} gives
$$
n_{\pm}\big(\varepsilon/2;\re T_{\sf bound}(\lambda)\big)=\O(1)
\quad{\rm as}\quad\lambda\to\pm m.
$$
Moreover, we know from \eqref{limit_im} that
$$
\lim_{\lambda\to\pm m}\big\|\im T_{\sf bound}(\lambda)\big\|_1=0.
$$
So the claim follows from the estimates \eqref{gniarf} and Formula \eqref{lhs}
\end{proof}

%-----------------------------------------------------------------------------------------------
\subsection{The case $\boldsymbol{|\lambda|<m}$}\label{inside}
%-----------------------------------------------------------------------------------------------

In this section we prove asymptotic estimates for $\xi(\lambda;H,H_\pm)$ as $\lambda\to\pm m$
with $|\lambda|<m$. We start with a corollary of Proposition \ref{empanada}, which follows
from the fact that $\im T_{\sf div}(\lambda)=0$ and
$\re T_{\sf div}(\lambda)=T_{\sf div}(\lambda)$ for $\lambda\in(-m,m)$.

\begin{Corollary}\label{jet-lag}
Let $V$ satisfy Assumption \ref{assumption2}. Then the estimates
$$
n_\pm\big(1+\varepsilon;T_{\sf div}(\lambda)\big)+\O(1)
\le\mp\xi(\lambda;H_\mp,H_0)
\le n_\pm\big(1-\varepsilon;T_{\sf div}(\lambda)\big)+\O(1)
$$
hold as $\lambda\to\pm m$, $|\lambda|<m$, for each $\varepsilon\in(0,1)$.
\end{Corollary}

Define the bounded operators $K_\pm:\H\to\ltwo(\R^2;\C^4)$ by
\begin{align*}
(K_+\varphi)(x_\perp)&:=\int_{\R^3}\d x_\perp'\d x_3'\,p(x_\perp,x_\perp')
\(\begin{smallmatrix}
1 & 0 & 0 & 0\\
0 & 0 & 0 & 0\\
0 & 0 & 0 & 0\\
0 & 0 & 0 & 0
\end{smallmatrix}\)
V^{1/2}(x_\perp',x_3')\varphi(x_\perp',x_3'),\\
(K_-\varphi)(x_\perp)&:=\int_{\R^3}\d x_\perp'\d x_3'\,p(x_\perp,x_\perp')
\(\begin{smallmatrix}
0 & 0 & 0 & 0\\
0 & 0 & 0 & 0\\
0 & 0 & 1 & 0\\
0 & 0 & 0 & 0
\end{smallmatrix}\)
V^{1/2}(x_\perp',x_3')\varphi(x_\perp',x_3'),
\end{align*}
where $p(\;\!\cdot\;\!,\;\!\cdot\;\!)$ is the integral kernel of the projection $p$.
One shows easily that $K_\pm^*:\ltwo(\R^2;\C^4)\to\H$ are given by
\begin{align*}
(K_+^*\psi)(x_\perp,x_3)&=V^{1/2}(x_\perp,x_3)
\(\begin{smallmatrix}
1 & 0 & 0 & 0\\
0 & 0 & 0 & 0\\
0 & 0 & 0 & 0\\
0 & 0 & 0 & 0
\end{smallmatrix}\)
(p\psi)(x_\perp),\\
(K_-^*\psi)(x_\perp,x_3)&=V^{1/2}(x_\perp,x_3)
\(\begin{smallmatrix}
0 & 0 & 0 & 0\\
0 & 0 & 0 & 0\\
0 & 0 & 1 & 0\\
0 & 0 & 0 & 0
\end{smallmatrix}\)
(p\psi)(x_\perp),
\end{align*}
and that
$$
O_+(\lambda):=\textstyle\frac12\big(\frac{m+\lambda}{m-\lambda}\big)^{1/2}K_+^*K_+
\qquad{\rm and}\qquad
O_-(\lambda):=\textstyle-\frac12\big(\frac{m-\lambda}{m+\lambda}\big)^{1/2}K_-^*K_-
$$
belong to $S_2(\H)$ for each $\lambda\in(-m,m)$.

In the next proposition we show that the functions
$n_\pm\big(\;\!\cdot\;\!;T_{\sf div}(\lambda)\big)$ as $\lambda\to\pm m$, $|\lambda|<m$,
can be bounded, up to $\O(1)$ terms, from below and from above by expressions involving
$O_\pm(\lambda)$.

\begin{Proposition}\label{PoppaChubby}
Let $V$ satisfy Assumption \ref{assumption2}. Then the estimates
\begin{align}
n_+\big((1+\varepsilon)s;O_+(\lambda)\big)+\O(1)
&\le n_+\big(s;T_{\sf div}(\lambda)\big)
\le n_+\big((1-\varepsilon)s;O_+(\lambda)\big)+\O(1),\label{positive1}\\
\O(1)&\le n_-\big(s;T_{\sf div}(\lambda)\big)\le \O(1),\label{positive2}
\end{align}
hold as $\lambda\nearrow m$, for each $\varepsilon\in(0,1)$ and $s>0$,
and the estimates
\begin{align}
\O(1)&\le n_+\big(s;T_{\sf div}(\lambda)\big)\le \O(1),\label{positive3}\\
n_-\big((1+\varepsilon)s;O_-(\lambda)\big)+\O(1)
&\le n_-\big(s;T_{\sf div}(\lambda)\big)
\le n_-\big((1-\varepsilon)s;O_-(\lambda)\big)+\O(1),\label{positive4}
\end{align}
hold as $\lambda\searrow-m$, for each $\varepsilon\in(0,1)$ and $s>0$.
\end{Proposition}

\begin{proof}
We only give the proof of \eqref{positive1}-\eqref{positive2}, since the proof of
\eqref{positive3}-\eqref{positive4} is similar. In point (i) below we show that the
difference $T_{\sf div}(\lambda)-O_+(\lambda)$ can be approximated in norm, as
$\lambda\nearrow m$, by a compact operator independent of $\lambda$. Then
we prove \eqref{positive1}-\eqref{positive2} in point (ii) by using this result.

(i) Let $\lambda\in(-m,m)$ and take $\nu'\in(3,\nu)$. A direct calculation shows that
\begin{equation}\label{TminusO}
T_{\sf div}(\lambda)-O_+(\lambda)
=\widetilde M\big(G_{\nu-\nu'}\otimes J^{(\lambda)}_{\nu'}\big)
\(\begin{smallmatrix}
(\lambda+m) & 0 & 0 & 0\\
0 & 0 & 0 & 0\\
0 & 0 & (\lambda-m) & 0\\
0 & 0 & 0 & 0
\end{smallmatrix}\)
\widetilde M+O_-(\lambda),
\end{equation}
where $J^{(\lambda)}_{\nu'}:\ltwo(\R)\to\ltwo(\R)$ is given by
$$
\big(J^{(\lambda)}_{\nu'}\psi\big)(x_3):=-\langle x_3\rangle^{-\nu'/2}\int_\R\d x_3'
\frac{\e^{-\frac12\sqrt{m^2-\lambda^2}|x_3-x_3'|}}{\sqrt{m^2-\lambda^2}}\,
\sinh\Big(\frac{\sqrt{m^2-\lambda^2}|x_3-x_3'|}2\Big)\langle x_3'\rangle^{-\nu'/2}\psi(x_3'),
$$
and
\begin{align}
\widetilde M&:=V^{1/2}\langle Q_\perp\rangle^{(\nu-\nu')/2}
\langle Q_3\rangle^{\nu'/2},\label{Mtilde}\\
G_{\nu-\nu'}&:=\langle Q_\perp\rangle^{-(\nu-\nu')/2}p
\langle Q_\perp\rangle^{-(\nu-\nu')/2}.\label{G}
\end{align}
The operator $\widetilde M$ is bounded due to Assumption \ref{assumption2}, $G_{\nu-\nu'}$
is compact in $\ltwo(\R^2;\C^4)$ due to Lemma \ref{lem_rai}, and
$O_-(\lambda)$ satisfies
\begin{equation}\label{omoins}
\lim_{\lambda\to m,\,|\lambda|<m}
\big\|O_-(\lambda)\big\|_2=0.
\end{equation}
Define
\begin{equation}\label{Tplusminus}
T_\pm:=\widetilde M\big(G_{\nu-\nu'}\otimes J^{(m)}_{\nu'}\big)
\(\begin{smallmatrix}
(m\pm m) & 0 & 0 & 0\\
0 & 0 & 0 & 0\\
0 & 0 & -(m\mp m) & 0\\
0 & 0 & 0 & 0
\end{smallmatrix}\)
\widetilde M,
\end{equation}
with $J^{(m)}_{\nu'}:\ltwo(\R)\to\ltwo(\R)$ given by
$$
\big(J^{(m)}_{\nu'}\psi\big)(x_3):=-\12\langle x_3\rangle^{-\nu'/2}\int_\R\d x_3'\,
|x_3-x_3'|\langle x_3'\rangle^{-\nu'/2}\psi(x_3').
$$
Since $\nu'>3$, $J^{(m)}_{\nu'}$ belongs to $S_2[\ltwo(\R)]$, and $T_\pm$ is compact
in $\H$. Moreover, by using Lebesgue's dominated convergence theorem, one shows that
$$
\lim_{\lambda\to\pm m,\,|\lambda|<m}
\big\|J^{(m)}_{\nu'}-J^{(\lambda)}_{\nu'}\big\|^2_2=0.
$$
This, together with \eqref{TminusO}, \eqref{omoins} and \eqref{Tplusminus}, implies
that
\begin{equation}\label{CaptainKirk}
\lim_{\lambda\nearrow m}\big\|T_{\sf div}(\lambda)-O_+(\lambda)-T_+\big\|=0.
\end{equation}

(ii) Take $\lambda\in(-m,m)$, $\varepsilon\in(0,1)$, and $s>0$. Using the Weyl
inequalities \eqref{Weyl} we get
\begin{align*}
n_\pm\big((1+\varepsilon)s;O_+(\lambda)\big)
-n_\mp\big(\varepsilon s;T_{\sf div}(\lambda)-O_+(\lambda)\big)
&\le n_\pm\big(s;T_{\sf div}(\lambda)\big)\\
&\le n_\pm\big((1-\varepsilon)s;O_+(\lambda)\big)
+n_\pm\big(\varepsilon s;T_{\sf div}(\lambda)-O_+(\lambda)\big).
\end{align*}
Now we have $n_-\big(t;O_+(\lambda)\big)=0$ for each $t>0$ and $\lambda\in(-m,m)$,
since $O_+(\lambda)$ is a positive operator. So, to prove
\eqref{positive1}-\eqref{positive2}, it is sufficient to show that
$
n_\pm\big(\varepsilon s;T_{\sf div}(\lambda)-O_+(\lambda)\big)=\O(1)
$
as $\lambda\nearrow m$, for each $\varepsilon\in(0,1)$ and $s>0$. Let
$t>0$ be fixed. Then we know from \eqref{CaptainKirk} that we can chose
$\lambda_+\in(-m,m)$, close enough to $m$, so that
$\big\|T_{\sf div}(\lambda_+)-O_+(\lambda_+)-T_+\big\|<t/2$. Thus, using again the
Weyl inequalities, we get
$$
n_\pm\big(t;T_{\sf div}(\lambda_+)-O_+(\lambda_+)\big)
\le n_\pm\big(t/2;T_{\sf div}(\lambda_+)-O_+(\lambda_+)-T_+\big)+n_\pm\big(t/2;T_+\big)
=n_\pm\big(t/2;T_+\big).
$$
Since the r.h.s. is independent of $\lambda_+$ we have shown that
$n_\pm\big(t;T_{\sf div}(\lambda)-O_+(\lambda)\big)=\O(1)$ as $\lambda\nearrow m$.
This concludes the proof of \eqref{positive1}-\eqref{positive2}.
\end{proof}

We show now that the counting functions $n_\pm\big(\;\!\cdot\;\!;O_\pm(\lambda)\big)$
in Proposition \ref{PoppaChubby} can be rewritten in terms of Berezin-Toeplitz type operators.
Define for each $\lambda\in(-m,m)$
$$
\omega_+(\lambda)
:=\textstyle\frac12\big(\frac{m+\lambda}{m-\lambda}\big)^{1/2}pW_+p
\qquad{\rm and}\qquad
\omega_-(\lambda)
:=\textstyle-\frac12\big(\frac{m-\lambda}{m+\lambda}\big)^{1/2}pW_-p,
$$
where the functions $W_\pm:\R^2\to\R$ are given by
\begin{equation}\label{BigOmegas}
W_+(x_\perp):=\int_\R\d x_3\,V_{11}(x_\perp,x_3)
\qquad{\rm and}\qquad
W_-(x_\perp):=\int_\R\d x_3\,V_{33}(x_\perp,x_3).
\end{equation}
Under the condition \eqref{a_decay} one has
$$
0\le W_\pm(x_\perp)\le{\rm Const.}\;\!\langle x_\perp\rangle^{-\nu+1}
\quad\hbox{for all }x_\perp\in\R^2,
$$
and $\omega_\pm(\lambda)\in S_1[\ltwo(\R^2)]$ if $V$ satisfies Assumption
\ref{first_decay} (see Lemma \ref{lem_rai}). Moreover, one has the
following.

\begin{Proposition}\label{spec_AB}
Let $V$ satisfy Assumption \ref{assumption1}. Then we have for each $\lambda\in(-m,m)$
and $s>0$
\begin{equation}\label{O=omega}
n_\pm\big(s;O_\pm(\lambda)\big)=n_\pm\big(s;\omega_\pm(\lambda)\big).
\end{equation}
\end{Proposition}

\begin{proof}
Given $s>0$ and two separable Hilbert spaces $\H_1,\H_2$, one has
\begin{equation}\label{BB*=B*B}
n_\pm\big(s;B^*B\big)=n_\pm\big(s;BB^*\big)
\end{equation}
for any $B\in\B(\H_1,\H_2)$ such that $B^*B\in S_\infty(\H_1)$. Moreover, one can
easily check that
$$
K_+K_+^*=
\(\begin{smallmatrix}
1 & 0 & 0 & 0\\
0 & 0 & 0 & 0\\
0 & 0 & 0 & 0\\
0 & 0 & 0 & 0
\end{smallmatrix}\)
pW_+p
\qquad{\rm and}\qquad
K_-K_-^*=
\(\begin{smallmatrix}
0 & 0 & 0 & 0\\
0 & 0 & 0 & 0\\
0 & 0 & 1 & 0\\
0 & 0 & 0 & 0
\end{smallmatrix}\)
pW_-p.
$$
Thus
$$
\textstyle
n_+\big(s;O_+(\lambda)\big)
=n_+\bigg(s;\frac12\big(\frac{m+\lambda}{m-\lambda}\big)^{1/2}
\(\begin{smallmatrix}
1 & 0 & 0 & 0\\
0 & 0 & 0 & 0\\
0 & 0 & 0 & 0\\
0 & 0 & 0 & 0
\end{smallmatrix}\)
pW_+p\bigg)
=n_+\Big(s;\frac12\big(\frac{m+\lambda}{m-\lambda}\big)^{1/2}pW_+p\Big)
=n_+\big(s;\omega_+(\lambda)\big).
$$
The proof of the second equality in \eqref{O=omega} is similar.
\end{proof}

The next theorem is direct consequence of Corollary \ref{jet-lag} and Propositions
\ref{PoppaChubby}-\ref{spec_AB}.

\begin{Theorem}\label{thm<}
Let $V$ satisfy Assumption \ref{assumption2}. Then one has for each $\varepsilon\in(0,1)$
\begin{align}
\O(1)&\le\xi(\lambda;H_+,H_0)\le \O(1)\label{ineq1}\\
-n_+\big(1-\varepsilon;\omega_+(\lambda)\big)+\O(1)
&\le\xi(\lambda;H_-,H_0)
\le-n_+\big(1+\varepsilon;\omega_+(\lambda)\big)+\O(1)\label{ineq2}
\end{align}
as $\lambda\nearrow m$, and
\begin{align}
n_-\big(1+\varepsilon;\omega_-(\lambda)\big)+\O(1)
&\le\xi(\lambda;H_+,H_0)
\le n_-\big(1-\varepsilon;\omega_-(\lambda)\big)+\O(1)\label{ineq3}\\
\O(1)&\le\xi(\lambda;H_-,H_0)\le \O(1)\label{ineq4}
\end{align}
as $\lambda\searrow-m$.
\end{Theorem}

\begin{Remark}
The inequalities \eqref{ineq1} together with Remark \eqref{int_eigen} imply that the
eigenvalues of $H_0+V$ in $(-m,m)$ near $+m$ (if any) do not accumulate at $+m$. On
the other hand the inequalities \eqref{ineq2} tell us that the number of eigenvalues
of $H_0-V$ in $(-m,m)$ near $\lambda=+m$ scales, up to $\O(1)$ terms, as
$$
\textstyle
n_+\big(s;\omega_+(\lambda)\big)\equiv\rank E^{pW_+p}
\Big(\Big(s\big(\frac{m-\lambda}{m+\lambda}\big)^{1/2},\infty\Big)\Big)
$$
with $s\approx2$. Accordingly, the problem of counting the number of eigenvalues of
$H_0-V$ in $(-m,m)$ near $+m$ reduces to the problem of counting the number of eigenvalues
of the positive Berezin-Toeplitz type operator $pW_+p$ near $0$. The inequalities
\eqref{ineq3}-\eqref{ineq4} lead to similar conclusions on the number
of eigenvalues of $H_0\pm V$ in $(-m,m)$ near $-m$.

One can compare these results with the results of \cite{Coj06} and \cite{IM99} on the
finiteness in $(-m,m)$ of the discrete spectrum of the Dirac operator perturbed by a
matrix potential $Q\equiv\{Q_{jk}(x)\}_{j,k=1}^4$. In Corollary 2.2 of \cite{Coj06}, the
author shows that the spectrum in $(-m,m)$ of the Dirac operator perturbed by $Q$ is
finite if the $2\times2$ diagonal blocks of $Q$ are of order $\O\big(|x|^{-2-\delta}\big)$
and the anti-diagonal blocks are of order $\O\big(|x|^{-1-\delta}\big)$, for some
$\delta>0$ as $|x|\to\infty$. In Corollary 2.1 of \cite{IM99}, the authors show that the
Dirac operator perturbed by $\gamma Q$, with $|\gamma|$ small enough and
$$
\big|Q_{jk}(x)\big|\le\langle x\rangle^{-2},\qquad j,k\in\{1,2,3,4\},
$$
does not have any point spectrum. Therefore, in our case where $Q=-\alpha_1a_1-\alpha_2a_2+V$,
we would not have had any accumulation of eigenvalues in $(-m,m)$ if we would have imposed
such decay assumptions on the magnetic part $-\alpha_1a_1-\alpha_2a_2$ of the perturbation.
\end{Remark}

As seen in Theorem \ref{thm<} the behaviour of the function $\xi(\;\!\cdot\;\!;H_\pm,H_0)$
in $(-m,m)$ depends on the distribution of
eigenvalues of the trace class operator $pW_\mp p$. In our next proposition we
shall exhibit different types of behaviours depending on the choice of the
functions $V_{11}$ and $V_{33}$ appearing in $W_\pm$. For that purpose, we first have
to recall some technical results taken from \cite{Rai09}, \cite{Rai03} and \cite{RW02}.

In the first lemma, an integrated density of states (IDS) for the operator $H_\perp^-$
in $\ltwo(\R^2)$ is defined as follows (see \eg \cite{DIM01,HLMW01}): Let $\chi_{T,x_\perp}$
be the characteristic function of the square $x_\perp+\big(-\frac T2,\frac T2\big)^2$, with
$x_\perp\in\R^2$ and $T>0$. Then a non-increasing function $\varrho:[0,\infty)\to\R$ is
called IDS for the  operator $H_\perp^-$ if for each $x_\perp\in\R^2$ it satisfies
$$
\varrho(\lambda)
=\lim_{T\to\infty}T^{-2}\tr\big[\chi_{T,x_\perp}(Q_\perp)
E^{H_\perp^-}\big((-\infty,\lambda)\big)\chi_{T,x_\perp}(Q_\perp)\big]
$$
for each point $\lambda\in\R$ of continuity of $\varrho$.

\begin{Lemma}[Lemma 3.3 of \cite{Rai09}]\label{tec_Rai1}
Let $U\in C^1(\R^2)$ satisfy
$$
0\le U(x_\perp)\le{\rm Const.}\;\!\langle x_\perp\rangle^{-\alpha}
\qquad\hbox{and}\qquad
\big|(\nabla U)(x_\perp)\big|\le{\rm Const.}\;\!\langle x_\perp\rangle^{-\alpha-1}
$$
for all $x\in\R^2$ and some $\alpha>0$. Assume moreover that
\begin{enumerate}
\item[$\bullet$] $U(x_\perp)=u\big(\frac{x_\perp}{|x_\perp|}\big)\big(1+o(1)\big)$ as
$|x_\perp|\to\infty$, where $u$ is a continuous function on $\mathbb S^1$ which does
not vanish identically,
\item[$\bullet$] $b$ is an admissible magnetic field,
\item[$\bullet$] there exists an IDS $\varrho_b$ for the operator $H_\perp^-$.
\end{enumerate}
Then we have
$$
n_+\big(s;pUp\big)
=\frac{b_0}{2\pi}\big|\big\{x_\perp\in\R^2\mid U(x_\perp)>s\big\}\big|\big(1+o(1)\big)
=\Psi_\alpha(s;u,b_0)\big(1+o(1)\big)\quad\hbox{as}\quad s\searrow0,
$$
where $|\;\!\cdot\;\!|$ denotes the Lebesgue measure, and
\begin{equation}\label{psi_alpha}
\Psi_\alpha(s;u,b_0):=\frac{s^{-2/\alpha}b_0}{4\pi}
\int_{\mathbb S^1}\d\vartheta\,u(\vartheta)^{2/\alpha},\qquad s>0.
\end{equation}
\end{Lemma}

\begin{Lemma}[Lemma 3.4 of \cite{Rai09}]\label{tec_Rai2}
Let $0\le U\in\linf(\R^2)$. Assume that
$$
\ln\big(U(x_\perp)\big)
=-\eta|x_\perp|^{2\beta}\big(1+o(1)\big)\quad\hbox{as}\quad|x_\perp|\to\infty,
$$
for some $\eta,\beta>0$. Let $b$ be an admissible magnetic field. Then we have
$$
n_+\big(s;pUp\big)=\Phi_\beta(s,\eta,b_0)\big(1+o(1)\big)
\quad\hbox{as}\quad s\searrow0,
$$
where
\begin{equation}\label{phi_beta}
\Phi_\beta(s,\eta,b_0):=
\begin{cases}
\frac{b_0}{2\eta^{1/\beta}}\;\!|\ln(s)|^{1/\beta} & \hbox{if}~~\beta\in(0,1),\\
\frac1{\ln(1+2\eta/b_0)}\;\!|\ln(s)| & \hbox{if}~~\beta=1,\\
\frac\beta{\beta-1}\big(\ln|\ln(s)|\big)^{-1}|\ln(s)| & \hbox{if}~~\beta>1,
\end{cases}
\qquad s\in(0,\e^{-1}).
\end{equation}
\end{Lemma}

\begin{Lemma}[Lemma 3.5 of \cite{Rai09}]\label{tec_Rai3}
Let $0\le U\in\linf(\R^2)$. Assume that the support of $U$ is compact, and that
there exists a constant $\textsc c>0$ such that $U\ge\textsc c$ on an open
non-empty subset of $\R^2$. Let $b$ be an admissible magnetic field. Then we have
$$
n_+\big(s;pUp\big)=\Phi_\infty(s)\big(1+o(1)\big)\quad\hbox{as}\quad s\searrow0,
$$
where
\begin{equation}\label{phi_infty}
\Phi_\infty(s):=\big(\ln|\ln(s)|\big)^{-1}|\ln(s)|,\qquad s\in(0,\e^{-1}).
\end{equation}
\end{Lemma}

Combining Theorem \ref{thm<} with Lemmas \ref{tec_Rai1}-\ref{tec_Rai3} we obtain
the behaviour of $\xi(\lambda;H_\pm,H_0)$ as $|\lambda|\to m$, $|\lambda|<m$,
when the functions $W_\pm$ admit a power-like or exponential decay at infinity,
or when they have a compact support.

\begin{Proposition}\label{in_gap}
Let $V$ satisfy Assumption \ref{assumption2}.
\begin{enumerate}
\item[(a)] Assume that the hypotheses of Lemma \ref{tec_Rai1} hold with
$U_\pm=W_\pm$ and $\alpha=\nu-1$. Then we have
$$
\xi(\lambda;H_-,H_0)
=\textstyle-\Psi_{\nu-1}\Big(2\big(\frac{m-\lambda}{m+\lambda}\big)^{1/2};u_+,b_0\Big)
\big(1+o(1)\big)\quad\hbox{as}\quad\lambda\nearrow m,
$$
and
$$
\xi(\lambda;H_+,H_0)
=\textstyle\Psi_{\nu-1}\Big(2\big(\frac{m+\lambda}{m-\lambda}\big)^{1/2};u_-,b_0\Big)
\big(1+o(1)\big)\quad\hbox{as}\quad\lambda\searrow-m,
$$
with $\Psi_{\nu-1}$ given by Equation \eqref{psi_alpha}.

\item[(b)] Assume that the hypotheses of Lemma \ref{tec_Rai2} hold with
$U_\pm=W_\pm$. Then we have
$$
\xi(\lambda;H_-,H_0)=\textstyle-\Phi_{\beta_+}
\Big(2\big(\frac{m-\lambda}{m+\lambda}\big)^{1/2};\eta_+,b_0\Big)\big(1+o(1)\big)
\quad\hbox{as}\quad\lambda\nearrow m,
$$
and
$$
\xi(\lambda;H_+,H_0)=\textstyle\Phi_{\beta_-}
\Big(2\big(\frac{m+\lambda}{m-\lambda}\big)^{1/2};\eta_-,b_0\Big)\big(1+o(1)\big)
\quad\hbox{as}\quad\lambda\searrow-m,
$$
with $\beta_\pm\in(0,\infty)$ and $\Phi_{\beta_\pm}$ given by Equation \eqref{phi_beta}.

\item[(c)] Assume that the hypotheses of Lemma \ref{tec_Rai3} hold with
$U_\pm=W_\pm$. Then we have
$$
\xi(\lambda;H_-,H_0)=\textstyle-\Phi_\infty
\Big(2\big(\frac{m-\lambda}{m+\lambda}\big)^{1/2}\Big)\big(1+o(1)\big)
\quad\hbox{as}\quad\lambda\nearrow m,
$$
and
$$
\xi(\lambda;H_+,H_0)=\textstyle\Phi_\infty
\Big(2\big(\frac{m+\lambda}{m-\lambda}\big)^{1/2}\Big)\big(1+o(1)\big)
\quad\hbox{as}\quad\lambda\searrow-m,
$$
with $\Phi_\infty$ given by Equation \eqref{phi_infty}.
\end{enumerate}
\end{Proposition}

The estimates of Proposition \ref{in_gap} are similar to the ones of \cite[Cor.~3.6]{Rai09},
where the corresponding situation for magnetic Pauli operators is considered.

%-----------------------------------------------------------------------------------------------
\subsection{The case $\boldsymbol{|\lambda|>m}$}\label{outside}
%-----------------------------------------------------------------------------------------------

In this section we prove asymptotic estimates for $\xi(\lambda;H,H_\pm)$ as $\lambda\to\pm m$,
when $|\lambda|>m$. We start by showing an estimate for
$n_\pm\big(s;\re T_{\sf div}(\lambda)\big)$.

\begin{Proposition}\label{pebre}
Let $V$ satisfy Assumption \ref{assumption2}. Then the estimates
$$
n_\pm\big(s;\re T_{\sf div}(\lambda)\big)=\O(1)
\quad\hbox{as}\quad\lambda\to\pm m,~|\lambda|>m,
$$
hold for each $s>0$.
\end{Proposition}

\begin{proof}
Take $\lambda\in\R$ with $|\lambda|>m$, and let $\nu'\in(3,\nu)$. Then we have
$$
\re T_{\sf div}(\lambda)=\widetilde M\big(G_{\nu-\nu'}\otimes R^{(\lambda)}_{\nu'}\big)
\(\begin{smallmatrix}
(\lambda+m) & 0 & 0 & 0\\
0 & 0 & 0 & 0\\
0 & 0 & (\lambda-m) & 0\\
0 & 0 & 0 & 0
\end{smallmatrix}\)\widetilde M,
$$
with $\widetilde M$ and $G_{\nu-\nu'}$ as in \eqref{Mtilde}-\eqref{G}, and
$$
R^{(\lambda)}_{\nu'}:=\langle Q_3\rangle^{-\nu'/2}\re R(\lambda^2-m^2)
\langle Q_3\rangle^{-\nu'/2}.
$$
By using Lebesgue's dominated convergence theorem, one shows that
$$
\lim_{\lambda\to\pm m,\,|\lambda|>m}\big\|\re T_{\sf div}(\lambda)-T_\pm\big\|=0,
$$
with $T_\pm$ as in \eqref{Tplusminus}. So the claim can be proved as in point (ii)
of the proof of Proposition \ref{PoppaChubby}.
\end{proof}

The next result follows from applying Propositions \ref{empanada} and \ref{pebre},
the Weyl inequalities \eqref{Weyl} and the identities \cite[Sec.~5.4]{FR04}
\begin{equation}\label{id_arctan}
\int_\R\d\mu(t)\,n_\pm\big(s;tT\big)=\pi^{-1}\tr\arctan(s^{-1}T),\qquad s>0,
\end{equation}
where $T\in S_1(\H)$, $T=T^*\ge0$. We also use the fact that
$\sgn(\lambda)\im T_{\sf div}(\lambda)$ is a positive operator if $|\lambda|>m$.

\begin{Corollary}\label{Musashi}
Let $V$ satisfy Assumption \ref{assumption2}. Then the estimates
\begin{align*}
&\pi^{-1}\tr\arctan
\big[(1+\varepsilon)^{-1}\sgn(\lambda)\im T_{\sf div}(\lambda)\big]+\O(1)\\
&\le\mp\xi(\lambda;H_\mp,H_0)\\
&\le\pi^{-1}\tr\arctan
\big[(1-\varepsilon)^{-1}\sgn(\lambda)\im T_{\sf div}(\lambda)\big]+\O(1)
\end{align*}
hold as $\lambda\to\pm m$, $|\lambda|>m$, for each $\varepsilon\in(0,1)$.
\end{Corollary}

As in the case $|\lambda|<m$, we introduce auxiliary operators in order to express
the lower and upper bounds for $\mp\xi(\lambda;H_\mp,H_0)$ in terms of Berezin-Toeplitz
type operators. For $\lambda\in\R$ with $|\lambda|>m$, we define the operators
$K_{1,\lambda},K_{2,\lambda}:\H\to\ltwo(\R^2;\C^4)$ by
\begin{align*}
(K_{1,\lambda}\varphi)(x_\perp)
&:=\int_{\R^3}\d x_\perp'\d x_3'\,p(x_\perp,x_\perp')
\cos\big(x_3'\sqrt{\lambda^2-m^2}\big)
\(\begin{smallmatrix}
\sqrt{|\lambda+m|} & 0 & 0 & 0\\
0 & 0 & 0 & 0\\
0 & 0 & \sqrt{|\lambda-m|} & 0\\
0 & 0 & 0 & 0
\end{smallmatrix}\)
V^{1/2}(x_\perp',x_3')\varphi(x_\perp',x_3'),\\
(K_{2,\lambda}\varphi)(x_\perp)
&:=\int_{\R^3}\d x_\perp'\d x_3'\,p(x_\perp,x_\perp')
\sin\big(x_3'\sqrt{\lambda^2-m^2}\big)
\(\begin{smallmatrix}
\sqrt{|\lambda+m|} & 0 & 0 & 0\\
0 & 0 & 0 & 0\\
0 & 0 & \sqrt{|\lambda-m|} & 0\\
0 & 0 & 0 & 0
\end{smallmatrix}\)
V^{1/2}(x_\perp',x_3')\varphi(x_\perp',x_3').
\end{align*}
Direct calculations show that the adjoint operators
$K_{1,\lambda}^*,K_{2,\lambda}^*:\ltwo(\R^2;\C^4)\to\H$ are given by
\begin{align*}
(K_{1,\lambda}^*\psi)(x_\perp,x_3)
&=\cos\big(x_3\sqrt{\lambda^2-m^2}\big)V^{1/2}(x_\perp,x_3)
\(\begin{smallmatrix}
\sqrt{|\lambda+m|} & 0 & 0 & 0\\
0 & 0 & 0 & 0\\
0 & 0 & \sqrt{|\lambda-m|} & 0\\
0 & 0 & 0 & 0
\end{smallmatrix}\)
(p\psi)(x_\perp),\\
(K_{2,\lambda}^*\psi)(x_\perp,x_3)
&=\sin\big(x_3\sqrt{\lambda^2-m^2}\big)V^{1/2}(x_\perp,x_3)
\(\begin{smallmatrix}
\sqrt{|\lambda+m|} & 0 & 0 & 0\\
0 & 0 & 0 & 0\\
0 & 0 & \sqrt{|\lambda-m|} & 0\\
0 & 0 & 0 & 0
\end{smallmatrix}\)
(p\psi)(x_\perp),
\end{align*}
and that
$$
\sgn(\lambda)\im T_{\sf div}(\lambda)=\frac1{2\sqrt{\lambda^2-m^2}}
\big(K_{1,\lambda}^*K_{1,\lambda}+K_{2,\lambda}^*K_{2,\lambda}\big).
$$
This last equation can be written more compactly as
\begin{equation}\label{graou}
\sgn(\lambda)\im T_{\sf div}(\lambda)
=\frac1{2\sqrt{\lambda^2-m^2}}\,K_\lambda^*K_\lambda
\end{equation}
if we use the operator
$$
K_\lambda:\H\to\ltwo(\R^2;\C^8),\qquad K_\lambda\varphi:=
\begin{pmatrix}
K_{1,\lambda}\varphi\\
K_{2,\lambda}\varphi
\end{pmatrix},
$$
with adjoint
$$
K_\lambda^*:\ltwo(\R^2;\C^8)\to\H,
\qquad K_{\lambda}^*\begin{pmatrix}\psi_1\\\psi_2\end{pmatrix}
=K_{1,\lambda}^*\psi_1+K_{2,\lambda}^*\psi_2.
$$

For the next proposition we also need to introduce for each $\lambda\in\R$ with
$|\lambda|>m$ the positive operator $\Omega(\lambda):\ltwo(\R^2;\C^8)\to\ltwo(\R^2;\C^8)$
defined by
$$
\Omega(\lambda):=\frac1{2\sqrt{\lambda^2-m^2}}\,K_\lambda K_\lambda^*.
$$
A direct calculation shows that
$$
K_\lambda K_\lambda^*
=p\begin{pmatrix}
M_{1,\lambda} & M_{2,\lambda}\\
M_{2,\lambda} & M_{3,\lambda}
\end{pmatrix}p,
$$
where
\begin{align*}
M_{1,\lambda}(x_\perp)
&:=\int_\R\d x_3\,\cos^2\big(x_3\sqrt{\lambda^2-m^2}\big)
\(\begin{smallmatrix}
|\lambda+m|V_{11}(x_\perp,x_3) & 0
& \sqrt{\lambda^2-m^2}V_{13}(x_\perp,x_3) & 0\\
0 & 0 & 0 & 0\\
\sqrt{\lambda^2-m^2}V_{31}(x_\perp,x_3) & 0
& |\lambda-m|V_{33}(x_\perp,x_3) & 0\\
0 & 0 & 0 & 0
\end{smallmatrix}\),\\
M_{2,\lambda}(x_\perp)
&=\int_\R\d x_3\,\sin\big(x_3\sqrt{\lambda^2-m^2}\big)
\cos\big(x_3\sqrt{\lambda^2-m^2}\big)
\(\begin{smallmatrix}
|\lambda+m|V_{11}(x_\perp,x_3) & 0
& \sqrt{\lambda^2-m^2}V_{13}(x_\perp,x_3) & 0\\
0 & 0 & 0 & 0\\
\sqrt{\lambda^2-m^2}V_{31}(x_\perp,x_3) & 0
& |\lambda-m|V_{33}(x_\perp,x_3) & 0\\
0 & 0 & 0 & 0
\end{smallmatrix}\),\\
M_{3,\lambda}(x_\perp)
&=\int_\R\d x_3\,\sin^2\big(x_3\sqrt{\lambda^2-m^2}\big)
\(\begin{smallmatrix}
|\lambda+m|V_{11}(x_\perp,x_3) & 0
& \sqrt{\lambda^2-m^2}V_{13}(x_\perp,x_3) & 0\\
0 & 0 & 0 & 0\\
\sqrt{\lambda^2-m^2}V_{31}(x_\perp,x_3) & 0
& |\lambda-m|V_{33}(x_\perp,x_3) & 0\\
0 & 0 & 0 & 0
\end{smallmatrix}\).
\end{align*}
This implies that
$$
\textstyle
\|\Omega(\lambda)\|_1
\le\big(\frac{\lambda+m}{\lambda-m}\big)^{1/2}\big\|pW_+p\big\|_1
+\big(\frac{\lambda-m}{\lambda+m}\big)^{1/2}\big\|pW_-p\big\|_1,
$$
and thus $\Omega(\lambda)\in S_1[\ltwo(\R^2;\C^8)]$ if $V$ satisfies Assumption
\ref{assumption1}.

Next Proposition is a direct consequence of Equations \eqref{BB*=B*B} and
\eqref{graou}.

\begin{Proposition}
Let $V$ satisfy Assumption \ref{assumption1}. Then we have for each $\lambda\in\R$
with $|\lambda|>m$ and each $s>0$
$$
n_\pm\big(s;\sgn(\lambda)\im T_{\sf div}(\lambda)\big)
=n_\pm\big(s;\Omega(\lambda)\big).
$$
In particular, it follows by Equation \eqref{id_arctan} that
\begin{equation}\label{cedula}
\tr\arctan\big(s^{-1}\sgn(\lambda)\im T_{\sf div}(\lambda)\big)
=\tr\arctan\big(s^{-1}\Omega(\lambda)\big).
\end{equation}
\end{Proposition}

The combination of Corollary \ref{Musashi} and Equation \eqref{cedula} gives the following.

\begin{Theorem}\label{thm_ext}
Let $V$ satisfy Assumption \ref{assumption2}. Then one has for each $\varepsilon\in(0,1)$
\begin{align*}
\pm\pi^{-1}\tr\arctan\big[(1\pm\varepsilon)^{-1}\Omega(\lambda)\big]+\O(1)
\le\xi(\lambda;H_\pm,H_0)
\le\pm\pi^{-1}\tr\arctan\big[(1\mp\varepsilon)^{-1}\Omega(\lambda)\big]+\O(1)
\end{align*}
as $\lambda\to\pm m$, $|\lambda|>m$.
\end{Theorem}

\begin{Remark}\label{C_sym}
The fact that the operators $\omega_\pm(\lambda)$ and $\Omega(\lambda)$ in
Theorems \ref{thm<} and \ref{thm_ext} depend in a distinguished way on the
components $V_{11}$ and $V_{33}$ of $V$ is due to our initial assumption
$b_0>0$. Indeed, this choice implies that $\ker(H^-_\perp)$ is non trivial,
whereas $\ker(H^+_\perp)=\{0\}$. This lead us to introduce in Section
\ref{Sec_Dec} the projection $\P\equiv\diag(P,0,P,0)$, which put into light the
priviledged role of the components $V_{11}$ and $V_{33}$ of $V$.

The variation of $\xi(\lambda;H_\pm,H_0)$ under the change
$\lambda\mapsto-\lambda$ can be explained using the antinunitary transformation
of charge conjugation \cite[Sec.~1.4.6]{Tha92}
$$
C:\H\to\H,\qquad\varphi\mapsto U_C\overline\varphi,
$$
where $U_C:=i\beta\alpha_2$. Indeed, if we write $H(\vec a,\pm V)$ and $H_0(\vec a)$
for $H_\pm$ and $H_0$, then a direct calculation using the Lifshits-Krein trace
formula \eqref{eq_LK} shows that
$$
CH(\vec a,\pm V)C^{-1}=-H(-\vec a,\mp U_C\overline VU_C^*),
$$
which entails
$$
\xi\big(\lambda;H(\vec a,\pm V),H_0(\vec a)\big)
=-\xi\big(-\lambda;H(-\vec a,\mp U_C\overline VU_C^*),H_0(-\vec a)\big).
$$
This obviously explains why the overall sign of the spectral shift function is reversed
under the change $\lambda\mapsto-\lambda$. But it also explains why the roles of $V_{11}$
and $V_{33}$ are interchanged in the estimates. Indeed, the natural
projection corresponding to the vector potential $\vec a$ is $\P=\diag(P,0,P,0)$ since
we have $b_0>0$ for $\vec a$, whereas $\P':=\diag(0,P,0,P)$ is the natural choice for the
vector potential $-\vec a$ since we have $b_0<0$ for $-\vec a$. Now, one has
$$
\mp U_C\overline VU_C^*=\mp
\left(\begin{smallmatrix}
V_{44} & -\overline{V_{43}} & -\overline{V_{42}} & \overline{V_{41}}\\
-\overline{V_{34}} & V_{33} & \overline{V_{32}} & -\overline{V_{31}}\\
-\overline{V_{24}} & \overline{V_{23}} & V_{22} & -\overline{V_{21}}\\
\overline{V_{14}} & -\overline{V_{13}} & -\overline{V_{12}} & V_{11}
\end{smallmatrix}\right).
$$
So, the projection $\P$ which selects the components $\pm(V_{11},V_{33})$ of the potential
$\pm V$ is replaced, after the change $\lambda\mapsto-\lambda$, by the projection $\P'$
which selects the components $\mp(V_{33},V_{11})$ of the transformed potential
$\mp U_C\overline VU_C^*$.
\end{Remark}

For the next proposition we define for each $\lambda\in\R$ with $|\lambda|>m$ the
positive operator $\Omega^{(1)}(\lambda)$ in $\ltwo(\R^2;\C^8)$ given by
$$
\Omega^{(1)}(\lambda):=\frac1{2\sqrt{\lambda^2-m^2}}
\begin{pmatrix}
pM_\lambda p & 0\\
0 & 0
\end{pmatrix}
\qquad\hbox{where}\qquad
M_\lambda:=
\(\begin{smallmatrix}
|\lambda+m|W_+ & 0 & 0 &0\\
0 & 0 & 0 &0\\
0 & 0 & |\lambda-m|W_- &0\\
0 & 0 & 0 &0
\end{smallmatrix}\).
$$

\begin{Proposition}\label{soprole}
\begin{enumerate}
\item[(a)] Let $V$ satisfy Assumption \ref{assumption2} with $\nu\in(3,4]$. Then one has
for each $s>0$ and each $\delta\in\big(\frac{4-\nu}2,\frac12\big)$
$$
\tr\big\{\arctan\big[s^{-1}\Omega(\lambda)\big]
-\arctan\big[s^{-1}\Omega^{(1)}(\lambda)\big]\big\}
=O\big(|\lambda\mp m|^{-\delta}\big)
\quad\hbox{as}\quad\lambda\to\pm m,~|\lambda|>m.
$$

\item[(b)] Let $V$ satisfy Assumption \ref{assumption1} with $\nu_\perp>2$ and
$\nu_3>2$. Then one has for each $s>0$
\begin{equation}\label{choclito}
\tr\big\{\arctan\big[s^{-1}\Omega(\lambda)\big]
-\arctan\big[s^{-1}\Omega^{(1)}(\lambda)\big]\big\}=\O(1)
\quad\hbox{as}\quad\lambda\to\pm m,~|\lambda|>m.
\end{equation}
\end{enumerate}
\end{Proposition}

\begin{proof}
Points (a) and (b) are proved by using the Lifshits-Krein trace formula \eqref{eq_LK}
with $f(\lambda)=\arctan(\lambda)$, $\lambda\in\R$. We do not give the details, since
the argument is analogous to the one of \cite[Cor.~2.2]{FR04}.
\end{proof}

Note that if $V$ satisfy Assumption \ref{assumption2} with $\nu\in(3,4]$, we can choose
$\delta\in\big(\frac{4-\nu}2,\frac1{\nu-1}\big)$, and so Proposition \ref{soprole}.(a)
entails
\begin{equation}
\tr\big\{\arctan\big[s^{-1}\Omega(\lambda)\big]
-\arctan\big[s^{-1}\Omega^{(1)}(\lambda)\big]\big\}
=o\big(|\lambda\mp m|^{-\frac1{\nu-1}}\big)
\quad\hbox{as}\quad\lambda\to\pm m,~|\lambda|>m.\label{eq_rhume2}
\end{equation}
Moreover, if $V$ satisfy Assumption \ref{assumption2} with $\nu>4$, then it satisfies
Assumption \ref{assumption1} with $\nu_\perp>2$ and $\nu_3>2$, and, hence
\eqref{choclito} is valid. Finally, we have for $s>0$ and $|\lambda|>m$
\begin{equation}\label{eq_rhume1}
\tr\arctan\big[s^{-1}\Omega^{(1)}(\lambda)\big]
=\int_0^\infty\frac{\d t}{1+t^2}\,
n_+{\textstyle\Big(2st\big(\frac{\lambda-m}{\lambda+m}\big)^{1/2}};pW_+p\Big)
+\int_0^\infty\frac{\d t}{1+t^2}\,
n_+{\textstyle\Big(2st\big(\frac{\lambda+m}{\lambda-m}\big)^{1/2}};pW_-p\Big).
\end{equation}

Combining Equations \eqref{choclito}-\eqref{eq_rhume1}, Theorem \ref{thm_ext} and
Lemmas \ref{tec_Rai1}-\ref{tec_Rai3}, we get the following.

\begin{Corollary}\label{outside_gap}
Let $V$ satisfy Assumption \ref{assumption2}.
\begin{enumerate}
\item[(a)] Assume that the hypotheses of Lemma \ref{tec_Rai1} hold with
$U_\pm=W_\pm$ and $\alpha=\nu-1$. Then we have
$$
\xi(\lambda;H_-,H_0)
=-\frac1{2\cos\big(\pi/(\nu-1)\big)}\,\textstyle
\Psi_{\nu-1}\Big(2\big(\frac{\lambda-m}{\lambda+m}\big)^{1/2};u_+,b_0\Big)
\big(1+o(1)\big)\quad\hbox{as}\quad\lambda\searrow m,
$$
and
$$
\xi(\lambda;H_+,H_0)
=	\frac1{2\cos\big(\pi/(\nu-1)\big)}\,\textstyle
\Psi_{\nu-1}\Big(2\big(\frac{\lambda+m}{\lambda-m}\big)^{1/2};u_-,b_0\Big)
\big(1+o(1)\big)\quad\hbox{as}\quad\lambda\nearrow-m,
$$
with $\Psi_{\nu-1}$ given by Equation \eqref{psi_alpha}.

\item[(b)] Suppose that $V$ also satisfies \eqref{first_decay} with $\nu_\perp>2$
and $\nu_3>2$, and assume that the hypotheses of Lemma \ref{tec_Rai2} hold with
$U_\pm=W_\pm$. Then we have
$$
\xi(\lambda;H_-,H_0)
=-\frac12\,\textstyle
\Phi_{\beta_+}\Big(2\big(\frac{\lambda-m}{\lambda+m}\big)^{1/2};\eta_+,b_0\Big)
\big(1+o(1)\big)\quad\hbox{as}\quad\lambda\searrow m,
$$
and
$$
\xi(\lambda;H_-,H_0)
=\frac12\,\textstyle
\Phi_{\beta_-}\Big(2\big(\frac{m+\lambda}{m-\lambda}\big)^{1/2};\eta_-,b_0\Big)
\big(1+o(1)\big)\quad\hbox{as}\quad\lambda\nearrow-m,
$$
with $\beta_\pm\in(0,\infty)$ and $\Phi_{\beta_\pm}$ given by Equation \eqref{phi_beta}.

\item[(c)] Suppose that $V$ also satisfies \eqref{first_decay} with $\nu_\perp>2$
and $\nu_3>2$, and assume that the hypotheses of Lemma \ref{tec_Rai3} hold with
$U_\pm=W_\pm$. Then we have
$$
\xi(\lambda;H_-,H_0)=-\frac12\,\textstyle\Phi_\infty
\Big(2\big(\frac{m-\lambda}{m+\lambda}\big)^{1/2}\Big)\big(1+o(1)\big)
\quad\hbox{as}\quad\lambda\searrow m,
$$
and
$$
\xi(\lambda;H_+,H_0)=\frac12\,\textstyle\Phi_\infty
\Big(2\big(\frac{m+\lambda}{m-\lambda}\big)^{1/2}\Big)\big(1+o(1)\big)
\quad\hbox{as}\quad\lambda\nearrow-m,
$$
with $\Phi_\infty$ given by Equation \eqref{phi_infty}.
\end{enumerate}
\end{Corollary}

Putting together the results of Proposition \ref{in_gap} and Corollary \ref{outside_gap}, we
obtain the following.

\begin{Corollary}\label{Levinson}
Under the assumptions of Corollary \ref{outside_gap}.(a), we have
$$
\lim_{\varepsilon\searrow0}
\frac{\xi\big(m(1-\varepsilon)^{-1};H_-,H_0\big)}{\xi\big(m(1-\varepsilon);H_-,H_0\big)}
=\frac1{2\cos\big(\pi/(\nu-1)\big)}
=\lim_{\varepsilon\searrow0}
\frac{\xi\big(-m(1-\varepsilon)^{-1};H_+,H_0\big)}{\xi\big(-m(1-\varepsilon);H_+,H_0\big)}\,,
$$
and under the assumptions of Corollary \ref{outside_gap}.(b)-(c), we have
$$
\lim_{\varepsilon\searrow0}
\frac{\xi\big(m(1-\varepsilon)^{-1};H_-,H_0\big)}{\xi\big(m(1-\varepsilon);H_-,H_0\big)}
=\frac12
=\lim_{\varepsilon\searrow0}
\frac{\xi\big(-m(1-\varepsilon)^{-1};H_+,H_0\big)}{\xi\big(-m(1-\varepsilon);H_+,H_0\big)}\,.
$$
\end{Corollary}

%-----------------------------------------------------------------------------------------------
\section*{Acknowledgements}
%-----------------------------------------------------------------------------------------------

The author expresses his deep gratitude to Professor G. D. Raikov for suggesting him this
study. He also thanks him for the idea of using the decomposition \eqref{tortugo} for the free
Hamiltonian. This work was partially supported by the Chilean Science Fundation Fondecyt under
the Grant 1090008.

%-----------------------------------------------------------------------------------------------
\section{Appendix}\label{cond_trace}
\setcounter{equation}{0}
%-----------------------------------------------------------------------------------------------

We give in this appendix the proof of the inclusion \eqref{necessary4} for the class of
potentials $V$ given in Remark \ref{a_class}. We start with a technical lemma. We use the
notations $\alpha:=(\alpha_1,\alpha_2,\alpha_3)^{\sf T}$ and
$$
(\partial_\ell V):=\{(\partial_\ell V_{jk})\},\qquad
\nabla V:=(\partial_1V,\partial_2V,\partial_3V)^{\sf T},\qquad
(\partial_\ell\partial_mV):=\{(\partial_\ell\partial_mV_{jk})\}.
$$

\begin{Lemma}
Let $V$ be as in Remark \ref{a_class}. Then
\begin{enumerate}
\item[(a)] One has in $\B\big(\dom(H_0),\dom(H_0)^*\big)$ the equalities
\begin{align}
[H_0,H]
&=-i\alpha\cdot(\nabla V)+[\alpha_3,V]P_3+m[\beta,V]\label{eq_com1}\\
&=-i\alpha\cdot(\nabla V)+P_3[\alpha_3,V]+i[\alpha_3,(\partial_3V)]+m[\beta,V].\label{eq_com2}
\end{align}

\item[(b)] Let $z\in\R\setminus\{\sigma(H_0)\cup\sigma(H_\pm)\}$. Then there exist operators
$B_\pm\in\B(\H)$ such that
\begin{equation}\label{square}
R_\pm^2(z)=B_\pm H_0^{-2}\qquad\hbox{and}\qquad R_\pm^2(z)=H_0^{-2}B_\pm^*.
\end{equation}
\end{enumerate}
\end{Lemma}

\begin{proof}
(a) We know from Lemma 2.2(b) of \cite{RT04} that $\dom(H_0)\subset\dom(P_3)$. So each member
of Equations \eqref{eq_com1}-\eqref{eq_com2} belongs to $\B\big(\dom(H_0),\dom(H_0)^*\big)$.

Let $\varphi\in\dom(H_0)$, take a sequence $\{\varphi_n\}\subset C^\infty_0(\R^3;\C^4)$
such that $\displaystyle\lim_n\|\varphi_n-\varphi\|_{\dom(H_0)}=0$, and denote by
$\langle\;\!\cdot\;\!,\;\!\cdot\;\!\rangle_{1,-1}$ the anti-duality map between $\dom(H_0)$
and $\dom(H_0)^*$. Then
\begin{align}
\langle\varphi,[H_0,H]\varphi\rangle_{1,-1}
&\equiv\langle H_0\varphi,V\varphi\rangle-\langle V\varphi,H_0\varphi\rangle\nonumber\\
&=\lim_n
\big\langle\varphi_n,[\alpha_1(P_1-a_1)+\alpha_2(P_2-a_2)+\alpha_3P_3+\beta m,V]\varphi_n\big\rangle
\nonumber\\
&=\lim_n\big\langle\varphi_n,
\big\{-i\alpha\cdot(\nabla V)+[\alpha_3,V]P_3+m[\beta,V]\big\}\varphi_n\big\rangle
\label{two_pos}.
\end{align}
Since $\dom(H_0)\subset\dom(P_3)$, we also have
$\displaystyle\lim_n\|\varphi_n-\varphi\|_{\dom(P_3)}=0$, and thus
$$
\langle\varphi,[H_0,H]\varphi\rangle_{1,-1}
=\big\langle\varphi,
\big\{-i\alpha\cdot(\nabla V)+[\alpha_3,V]P_3+m[\beta,V]\big\}\varphi\big\rangle
$$
This proves \eqref{eq_com1}. Using \eqref{two_pos}, one also gets the equality \eqref{eq_com2}.

(b) In what follows, we omit the indices ``$\pm$" to simplify the notations and we write
$B_1,B_2,\ldots$ for elements of $\B(\H)$. Since $\dom(H)=\dom(H_0)$, we have
$$
R^2(z)
=B_1H_0^{-1}R(z)
=B_1R(z)H_0^{-1}+B_1\big[H_0^{-1},R(z)\big]
=B_2H_0^{-2}+B_1H_0^{-1}R(z)\big[H_0,H\big]R(z)H_0^{-1}.
$$
Now, one has
$$
R(z)\big[H_0,H\big]R(z)H_0^{-1}
=R(z)\big\{-i\alpha\cdot(\nabla V)+P_3[\alpha_3,V]+i[\alpha_3,(\partial_3V)]+m[\beta,V]\big\}
R(z)H_0^{-1}
=B_3H_0^{-2}
$$
due to Equation \eqref{eq_com2}, the equality $\dom(H)=\dom(H_0)$, and the inclusion
$\dom(H_0)\subset\dom(P_3)$. This, together with the preceding equation, implies the first
identity in \eqref{square}. The second identity follows from the first one by adjunction.
\end{proof}

\begin{Proposition}\label{diff_trace}
Take $z\in\R\setminus\{\sigma(H_0)\cup\sigma(H_\pm)\}$ and let $V$ be as in Remark
\ref{a_class}. Then we have
$$
R_\pm^3(z)-R_0^3(z)\in S_1(\H).
$$
\end{Proposition}

\begin{proof}
In what follows, we omit the indices ``$\pm$" to simplify the notations and we write
$B_1,B_2,\ldots$ for elements of $\B(\H)$. Differentiating twice the resolvent identity
$$
R(z)-R_0(z)=-R(z)VR_0(z)
$$
we find that
$$
R^3(z)-R_0^3(z)=-R(z)VR_0^3(z)-R^2(z)VR_0^2(z)-R^3(z)VR_0(z).
$$
So it is sufficient to show that each term on the r.h.s. belongs to $S_1(\H)$. This is
done in points (i), (ii) and (iii) below.

(i) For the term $R(z)VR_0^3(z)$, one has
\begin{equation}\label{premiere_eq}
R(z)VR_0^3(z)=R(z)R_0(z)VR_0^2(z)+R(z)[V,R_0(z)]R_0^2(z).
\end{equation}
Since $\dom(H)=\dom(H_0)$, one has
$$
R(z)R_0(z)VR_0^2(z)
=R(z)(H_0-z)R_0^2(z)VR_0^2(z)
=\big(B_1H_0^{-2}V^{1/2}\big)\big(V^{1/2}H_0^{-2}B_2\big).
$$
So, by \eqref{necessary2}, $R(z)R_0(z)VR_0^2(z)$ is the product of two Hilbert-Schmidt
operators, and thus belongs to $S_1(\H)$.

For the second term of \eqref{premiere_eq}, we have by \eqref{eq_com1}
$$
R(z)[V,R_0(z)]R_0^2(z)
=R(z)R_0(z)[H_0,H]R_0^3(z)
=B_1H_0^{-2}\big\{-i\alpha\cdot(\nabla V)+[\alpha_3,V]P_3+m[\beta,V]\big\}H_0^{-3}B_3.
$$
Due to the hypotheses on $V$ and $(\partial_jV)$, one can use \eqref{necessary2} to
write the first and third term as a product of two Hilbert-Schmidt operators. So it
only remains to show that $H_0^{-2}[\alpha_3,V]P_3H_0^{-3}$ belongs to $S_1(\H)$. For
this, we use the inclusion $\dom(H_0)\subset\dom(P_3)$ and the commutation of $P_3$
and $H_0^{-1}$ on $\dom(P_3)$ \cite[Lemma~2.2(b)]{RT04} to get
$$
H_0^{-2}[\alpha_3,V]P_3H_0^{-3}
=H_0^{-2}[\alpha_3,V]H_0^{-2}P_3H_0^{-1}
=H_0^{-2}[\alpha_3,V]H_0^{-2}B_4.
$$
This, together with \eqref{necessary2}, implies that $H_0^{-2}[\alpha_3,V]P_3H_0^{-3}$
belongs to $S_1(\H)$.

(ii) One can write $R^2(z)VR_0^2(z)$ as the product of two Hilbert-Schmidt operators by
using \eqref{square} and \eqref{necessary2}:
$$
R^2(z)VR_0^2(z)
=B_5H_0^{-2}VH_0^{-2}B_2
=\big(B_5H_0^{-2}V^{1/2}\big)\big(V^{1/2}H_0^{-2}B_2\big).
$$
Thus $R^2(z)VR_0^2(z)$ belongs to $S_1(\H)$.

(iii) For the term $R^3(z)VR_0(z)$ we have
$$
R^3(z)VR_0(z)
=R^2(z)VR(z)R_0(z)+R^2(z)[R(z),V]R_0(z).
$$
One shows that $R^2(z)VR(z)R_0(z)\in S_1(\H)$ as in point (ii). For the second term, we have
by \eqref{eq_com2} and \eqref{square}
\begin{align*}
R^2(z)[R(z),V]R_0(z)
&=R^3(z)[H,H_0]R(z)R_0(z)\\
&=B_5H_0^{-2}R(z)
\big\{i\alpha\cdot(\nabla V)-P_3[\alpha_3,V]-i[\alpha_3,(\partial_3V)]-m[\beta,V]\big\}
H_0^{-2}B_6.
\end{align*}
Due to the hypotheses on $V$ and $(\partial_jV)$, one can use \eqref{square} and
\eqref{necessary2} to write the first, third, and fourth term as a product of two
Hilbert-Schmidt operators. So it only remains to show that
$H_0^{-2}R(z)P_3[\alpha_3,V]H_0^{-2}$ belongs to $S_1(\H)$. Using \cite[Lemma~2.2(b)]{RT04}
and \eqref{square}, one gets
\begin{align*}
H_0^{-2}R(z)P_3[\alpha_3,V]H_0^{-2}
&=H_0^{-2}P_3R(z)[\alpha_3,V]H_0^{-2}+H_0^{-2}[R(z),P_3][\alpha_3,V]H_0^{-2}\\
&=P_3H_0^{-2}R(z)[\alpha_3,V]H_0^{-2}-iH_0^{-2}R(z)(\partial_3V)R(z)[\alpha_3,V]H_0^{-2}\\
&=B_7H_0^{-2}[\alpha_3,V]H_0^{-2}+B_8R(z)(\partial_3V)R(z)[\alpha_3,V]H_0^{-2}.
\end{align*}
The first term on the r.h.s. belongs to $S_1(\H)$, and for the second term we have by
\eqref{eq_com2} and \eqref{square}
\begin{align*}
&R(z)(\partial_3V)R(z)[\alpha_3,V]H_0^{-2}\\
&=(\partial_3V)R^2(z)[\alpha_3,V]H_0^{-2}+R(z)[(\partial_3V),H]R^2(z)[\alpha_3,V]H_0^{-2}\\
&=B_9H_0^{-2}[\alpha_3,V]H_0^{-2}+R(z)\big\{i\alpha\cdot[\nabla(\partial_3V)]
-P_3[\alpha_3,(\partial_3V)]\\
&\hspace{150pt}-i[\alpha_3,(\partial_3^2V)]-m[\beta,(\partial_3V)]+[(\partial_3V),V]\big\}
B_5H_0^{-2}[\alpha_3,V]H_0^{-2}.
\end{align*}
Due to the hypotheses on $V$, $(\partial_jV)$, and $(\partial_{j3}V)$, one can use
\eqref{necessary2} to show that the first, the second, the fourth, the fifth, and the
sixth term are trace class. For the third term we have to use \eqref{necessary2} and the
fact that $R(z)P_3$ extends to a bounded operator.
\end{proof}

\begin{Remark}
When the potential $V$ is scalar, the equations \eqref{eq_com1}-\eqref{eq_com2} reduce to
the single equality
$$
[H_0,H]=-i\alpha\cdot(\nabla V)
$$
in $\B\big(\dom(H_0),\dom(H_0)^*\big)$. So the calculations in points (i) and (iii) of
the proof of Proposition \ref{diff_trace} simplify accordingly, and we obtain the inclusion
$$
R_\pm^3(z)-R_0^3(z)\in S_1(\H)
$$
without assuming anything on the derivatives of $V$ of order $2$.
\end{Remark}

%-----------------------------------------------------------------------------------------------
%\bibliography{../bibliographie/bibliographie}

\begin{thebibliography}{10}

\bibitem{AHS78}
J.~Avron, I.~Herbst, and B.~Simon.
\newblock Schr\"odinger operators with magnetic fields. {I}. {G}eneral
  interactions.
\newblock {\em Duke Math. J.}, 45(4):847--883, 1978.

\bibitem{BG87}
A.~Berthier and V.~Georgescu.
\newblock On the point spectrum of {D}irac operators.
\newblock {\em J. Funct. Anal.}, 71(2):309--338, 1987.

\bibitem{BS99}
M.~Sh. Birman and T.~A. Suslina.
\newblock The periodic {D}irac operator is absolutely continuous.
\newblock {\em Integral Equations Operator Theory}, 34(4):377--395, 1999.

\bibitem{BG09}
N.~Boussaid and S.~Gol\'enia.
\newblock Limiting absorption principle for some long range perturbations of
  {D}irac systems at threshold energies.
\newblock preprint on \texttt{http://arxiv.org/abs/0906.1495}.

\bibitem{BMR93}
A.~Boutet~de Monvel-Berthier, D.~Manda, and R.~Purice.
\newblock Limiting absorption principle for the {D}irac operator.
\newblock {\em Ann. Inst. H. Poincar\'e Phys. Th\'eor.}, 58(4):413--431, 1993.

\bibitem{BPR04}
V.~Bruneau, A.~Pushnitski, and G.~Raikov.
\newblock Spectral shift function in strong magnetic fields.
\newblock {\em Algebra i Analiz}, 16(1):207--238, 2004.

\bibitem{BR99}
V.~Bruneau and D.~Robert.
\newblock Asymptotics of the scattering phase for the {D}irac operator: high
  energy, semi-classical and non-relativistic limits.
\newblock {\em Ark. Mat.}, 37(1):1--32, 1999.

\bibitem{Coj06}
P.~A. Cojuhari.
\newblock On the finiteness of the discrete spectrum of the {D}irac operator.
\newblock {\em Rep. Math. Phys.}, 57(3):333--341, 2006.

\bibitem{Dav07}
E.~B. Davies.
\newblock {\em Linear operators and their spectra}, volume 106 of {\em
  Cambridge Studies in Advanced Mathematics}.
\newblock Cambridge University Press, Cambridge, 2007.

\bibitem{DIM01}
S.~Doi, A.~Iwatsuka, and T.~Mine.
\newblock The uniqueness of the integrated density of states for the
  {S}chr\"odinger operators with magnetic fields.
\newblock {\em Math. Z.}, 237(2):335--371, 2001.

\bibitem{EL99}
W.~D. Evans and R.~T. Lewis.
\newblock Eigenvalue estimates in the semi-classical limit for {P}auli and
  {D}irac operators with a magnetic field.
\newblock {\em R. Soc. Lond. Proc. Ser. A Math. Phys. Eng. Sci.},
  455(1981):183--217, 1999.

\bibitem{FR04}
C.~Fern{\'a}ndez and G.~Raikov.
\newblock On the singularities of the magnetic spectral shift function at the
  {L}andau levels.
\newblock {\em Ann. Henri Poincar\'e}, 5(2):381--403, 2004.

\bibitem{GM01}
V.~Georgescu and M.~M{\u{a}}ntoiu.
\newblock On the spectral theory of singular {D}irac type {H}amiltonians.
\newblock {\em J. Operator Theory}, 46(2):289--321, 2001.

\bibitem{GM00}
F.~Gesztesy and K.~A. Makarov.
\newblock The {$\Xi$} operator and its relation to {K}rein's spectral shift
  function.
\newblock {\em J. Anal. Math.}, 81:139--183, 2000.

\bibitem{GM93}
A.~Grigis and A.~Mohamed.
\newblock Finitude des lacunes dans le spectre de l'op\'erateur de
  {S}chr\"odinger et de celui de {D}irac avec des potentiels \'electrique et
  magn\'etique p\'eriodiques.
\newblock {\em J. Math. Kyoto Univ.}, 33(4):1071--1096, 1993.

\bibitem{Hac93}
G.~Hachem.
\newblock Effet {Z}eeman pour un \'electron de {D}irac.
\newblock {\em Ann. Inst. H. Poincar\'e Phys. Th\'eor.}, 58(1):105--123, 1993.

\bibitem{HNW89}
B.~Helffer, J.~Nourrigat, and X.~P. Wang.
\newblock Sur le spectre de l'\'equation de {D}irac (dans ${\R^2}$ ou ${\R^3}$)
  avec champ magn\'etique.
\newblock {\em Ann. scient. \'Ec. Norm. Sup.}, 22:515--533, 1989.

\bibitem{HLMW01}
T.~Hupfer, H.~Leschke, P.~M{\"u}ller, and S.~Warzel.
\newblock Existence and uniqueness of the integrated density of states for
  {S}chr\"odinger operators with magnetic fields and unbounded random
  potentials.
\newblock {\em Rev. Math. Phys.}, 13(12):1547--1581, 2001.

\bibitem{IM99}
A.~Iftimovici and M.~M{\u{a}}ntoiu.
\newblock Limiting absorption principle at critical values for the {D}irac
  operator.
\newblock {\em Lett. Math. Phys.}, 49(3):235--243, 1999.

\bibitem{Ivr98}
V.~Ivrii.
\newblock {\em Microlocal analysis and precise spectral asymptotics}.
\newblock Springer Monographs in Mathematics. Springer-Verlag, Berlin, 1998.

\bibitem{Kla90}
M.~Klaus.
\newblock On the {L}evinson theorem for {D}irac operators.
\newblock {\em J. Math. Phys.}, 31(1):182--190, 1990.

\bibitem{Ma06}
Z.-Q. Ma.
\newblock The {L}evinson theorem.
\newblock {\em J. Phys. A}, 39(48):R625--R659, 2006.

\bibitem{MR03}
M.~Melgaard and G.~Rozenblum.
\newblock Eigenvalue asymptotics for weakly perturbed {D}irac and
  {S}chr\"odinger operators with constant magnetic fields of full rank.
\newblock {\em Comm. Partial Differential Equations}, 28(3-4):697--736, 2003.

\bibitem{Mur90}
G.~J. Murphy.
\newblock {\em {$C\sp *$}-algebras and operator theory}.
\newblock Academic Press Inc., Boston, MA, 1990.

\bibitem{Pus01}
A.~Pushnitski.
\newblock The spectral shift function and the invariance principle.
\newblock {\em J. Funct. Anal.}, 183(2):269--320, 2001.

\bibitem{Pus97}
A.~B. Pushnitski{\u\i}.
\newblock A representation for the spectral shift function in the case of
  perturbations of fixed sign.
\newblock {\em Algebra i Analiz}, 9(6):197--213, 1997.

\bibitem{Rai09}
G.~D. Raikov.
\newblock {L}ow {E}nergy {A}symptotics of the {SSF} for {P}auli {O}perators
  with {N}onconstant {M}agnetic {F}ields.
\newblock preprint on \texttt{http://arxiv.org/abs/0908.3704}.

\bibitem{Rai99}
G.~D. Raikov.
\newblock Eigenvalue asymptotics for the {D}irac operator in strong constant
  magnetic fields.
\newblock {\em Math. Phys. Electron. J.}, 5:Paper 2, 22 pp.\ (electronic),
  1999.

\bibitem{Rai03}
G.~D. Raikov.
\newblock Spectral asymptotics for the perturbed 2{D} {P}auli operator with
  oscillating magnetic fields. {I}. {N}on-zero mean value of the magnetic
  field.
\newblock {\em Markov Process. Related Fields}, 9(4):775--794, 2003.

\bibitem{RW02}
G.~D. Raikov and S.~Warzel.
\newblock Quasi-classical versus non-classical spectral asymptotics for
  magnetic {S}chr\"odinger operators with decreasing electric potentials.
\newblock {\em Rev. Math. Phys.}, 14(10):1051--1072, 2002.

\bibitem{RS79}
M.~Reed and B.~Simon.
\newblock {\em Methods of modern mathematical physics. {III}}.
\newblock Academic Press [Harcourt Brace Jovanovich Publishers], New York,
  1979.
\newblock Scattering theory.

\bibitem{RT04}
S.~Richard and R.~{Tiedra de Aldecoa}.
\newblock On perturbations of {D}irac operators with variable magnetic field of
  constant direction.
\newblock {\em J. Math. Phys.}, 45(11):4164--4173, 2004.

\bibitem{RT07}
S.~Richard and R.~Tiedra~de Aldecoa.
\newblock On the spectrum of magnetic {D}irac operators with {C}oulomb-type
  perturbations.
\newblock {\em J. Funct. Anal.}, 250(2):625--641, 2007.

\bibitem{Rob99}
D.~Robert.
\newblock Semiclassical asymptotics for the spectral shift function.
\newblock In {\em Differential operators and spectral theory}, volume 189 of
  {\em Amer. Math. Soc. Transl. Ser. 2}, pages 187--203. Amer. Math. Soc.,
  Providence, RI, 1999.

\bibitem{Saf01}
O.~Safronov.
\newblock Spectral shift function in the large coupling constant limit.
\newblock {\em J. Funct. Anal.}, 182(1):151--169, 2001.

\bibitem{SU09}
Y.~Saito and T.~Umeda.
\newblock Eigenfunctions at the threshold energies of magnetic {D}irac
  operators.
\newblock preprint on \texttt{http://arxiv.org/abs/0905.0961}.

\bibitem{Tha91}
B.~Thaller.
\newblock Dirac particles in magnetic fields.
\newblock In {\em Recent developments in quantum mechanics ({P}oiana {B}ra\c
  sov, 1989)}, volume~12 of {\em Math. Phys. Stud.}, pages 351--366. Kluwer
  Acad. Publ., Dordrecht, 1991.

\bibitem{Tha92}
B.~Thaller.
\newblock {\em The {D}irac Equation}.
\newblock Springer-Verlag, Berlin, 1992.

\bibitem{Yaf92}
D.~R. Yafaev.
\newblock {\em Mathematical scattering theory}, volume 105 of {\em Translations
  of Mathematical Monographs}.
\newblock American Mathematical Society, Providence, RI, 1992.
\newblock General theory, Translated from the Russian by J. R. Schulenberger.

\bibitem{Yaf05}
D.~R. Yafaev.
\newblock A trace formula for the {D}irac operator.
\newblock {\em Bull. London Math. Soc.}, 37(6):908--918, 2005.

\end{thebibliography}
%-----------------------------------------------------------------------------------------------

\end{document}